\pdfoutput=1
\documentclass[12pt,a4paper]{article}
\usepackage{amsmath,amssymb,amsthm}
\usepackage[english]{babel}
\usepackage[T1]{fontenc}
\usepackage[utf8]{inputenc}
\usepackage{csquotes}
\usepackage{url}
\usepackage[numbers, square]{natbib}
\usepackage{mathrsfs}
\usepackage{tikz}
\usetikzlibrary{arrows}
\usetikzlibrary{shapes.geometric,calc,positioning}
\usepackage{mathtools}
\usepackage{ stmaryrd }
\usepackage{subcaption}

\newcommand\defeq{\stackrel{\mathclap{\normalfont\mbox{\tiny def}}}{=}}

\theoremstyle{plain}
\newtheorem{thm}{Theorem}
\newtheorem{lem}[thm]{Lemma}
\newtheorem{prop}[thm]{Proposition}
\newtheorem{cor}[thm]{Corollary}
\newtheorem{defn}[thm]{Definition}

\theoremstyle{definition}

\theoremstyle{remark}
\newtheorem*{rem}{Remark}
\newtheorem*{note}{Note}

\title{\vspace{-2.0cm} How to Generate Pseudorandom Permutations Over Other Groups:
Even-Mansour and Feistel Revisited} 
\author{Hector B. Hougaard\footnote{This article is based on work done for my Master's Thesis at the University of Copenhagen. For more details on the thesis, contact me at aehogo@gmail.com.}}
\date{\vspace{-5ex}}

\begin{document}
\maketitle

\begin{abstract}
Recent results by Alagic and Russell have given some evidence that the Even-Mansour cipher may be secure against quantum adversaries with quantum queries, if considered over other groups than $(\mathbb{Z}/2)^n$. This prompts the question as to whether or not other classical schemes may be generalized to arbitrary groups and whether classical results still apply to those generalized schemes.

In this paper, we generalize the Even-Mansour cipher and the Feistel cipher. We show that Even and Mansour's original notions of secrecy are obtained on a one-key, group variant of the Even-Mansour cipher. We generalize the result by Kilian and Rogaway, that the Even-Mansour cipher is pseudorandom, to super pseudorandomness, also in the one-key, group case. Using a Slide Attack we match the bound found above. After generalizing the Feistel cipher to arbitrary groups we resolve an open problem of Patel, Ramzan, and Sundaram by showing that the $3$-round Feistel cipher over an arbitrary group is not super pseudorandom.

Finally, we generalize a result by Gentry and Ramzan showing that the Even-Mansour cipher can be implemented using the Feistel cipher as the public permutation. In this last result, we also consider the one-key case over a group and generalize their bound.
\end{abstract}

\section{Introduction}
In \citep{EM}, Even and Mansour introduced and proved security for the DES inspired block cipher scheme we now call the Even-Mansour (EM) scheme. Given a public permutation, $P$, over $n$-bit strings, with two different, random, secret, $n$-bit keys $k_1$ and $k_2$, a message $x\in \lbrace 0,1 \rbrace^n$ could be enciphered as
\begin{align*}
EM_{k_1,k_2}^P (x) = P(x \oplus k_1)\oplus k_2,
\end{align*}
with an obvious decryption using the inverse public permutation. The scheme was minimal, in the sense that they needed to XOR a key before and after the permutation, otherwise the remaining key could easily be found. As an improvement, Dunkelman, Keller, and Shamir \citep{DKS} showed that there was only a need for a single key and the scheme would still retain an indistinguishability from random, i.e. it was pseudorandom. As another consideration of block ciphers, the Feistel cipher construction of Luby and Rackoff \citep{LubyR} showed how to build pseudorandom permutations from pseudorandom functions. 

Eventually, Kuwakado and Morii showed that both the EM scheme \citep{BrokenEM} and the Feistel scheme \citep{BrokenFeistel} could be broken by quantum adversaries with quantum queries. Rather than discard these beautiful constructions entirely, Alagic and Russell \citep{Gorjan} considered whether it would be possible to define the two-key EM scheme over Abelian groups in order to retain security against quantum adversaries with quantum queries. What they showed was a security reduction to the Hidden Shift Problem, over certain groups, such as $\mathbb{Z}/2^n$ and $S_n$. This result inspires us to ask whether the EM and Feistel schemes can be generalized over all groups, and if so, whether or not we can get pseudorandomness in some model.

\subsection{Prior Work}
In extension of their simplification of the EM scheme, Dunkelman, Keller, and Shamir \citep{DKS} attacked the construction using variants of slide attacks in order to show that the security bound was optimal. They further considered other variants of the EM scheme, such as the Addition Even-Mansour with an Involution as the Permutation (two-keyed). Also Kilian and Rogaway \citep{Kilr} were inspired by DESX and EM to define their $FX$ construction, of which the EM scheme is a special case.

As referred to above, Kuwakado and Morii were able to break the EM scheme \citep{BrokenEM} and the $3$-round Feistel scheme \citep{BrokenFeistel} on $n$-bit strings, using Simon's algorithm, if able to query their oracle with a superposition of states. Kaplan et al. \citep{Kaplan}, using Kuwakado and Morii's results, showed how to break many classical cipher schemes, which in turn incited Alagic and Russell \citep{Gorjan}.

In their work with the Hidden Shift Problem, \citep{Gorjan} posit that a Feistel cipher construction over other groups than the bit strings might be secure against quantum adversaries with quantum queries. Many Feistel cipher variants exist, with different relaxations on the round functions, see for example \citep{NaorReingold} and \citep{PatelRamzanSundaram}, the latter of which also considered Feistel ciphers over other groups. Vaudenay \citep{Vaudenay} also considered Feistel ciphers over other groups in order to protect such ciphers against differential analysis attacks by what he called decorrelation.

Removed from the schemes considered below and with a greater degree of abstraction, Black and Rogaway \citep{BlackRogaway} consider ciphers over arbitrary domains. In general, on the question of the existence of quantum pseudorandom permutations, see \citep{Zhandry}.

\subsection{Summary of Results}
We work in the Random Oracle Model and consider groups $G$ in the family of finite groups, $\mathcal{G}$. We consider pseudorandom permutations, given informally as the following.

\begin{defn}
\textbf{[Informal]} A keyed permutation $P$ on a group $G$ is a \textbf{Pseudorandom Permutation (PRP)} on $G$ if it is indistinguishable from a random permutation for all probabilistic distinguishers having access to only polynomially many permutation-oracle queries.
\end{defn}

A \textbf{Super Pseudorandom Permutation (SPRP)} is a permutation where the distinguisher is given access to the inverse permutation-oracle as well.

We define the \textbf{Group Even-Mansour (EM) scheme} on $G$ to be the encryption scheme having the encryption algorithm
\begin{align*}
E_k(m) = P(m\cdot k) \cdot k,
\end{align*}
where $m\in G$ is the plaintext and $k\in G$ is the uniformly random key.

We define two problems for the Group Even-Mansour scheme: \textbf{Existential Forgery} (EFP) and \textbf{Cracking} (CP). In EFP, the adversary must eventually output a plaintext-ciphertext pair which satisfies correctness. In CP, the adversary is given a ciphertext and asked to find the corresponding plaintext.

It holds that for our Group EM scheme, the probability that an adversary succeeds in the EFP is polynomially bounded:
\begin{thm}\label{IntuitiveEFP}
\textbf{[Informal]} If $P$ is a uniformly random permutation on $G$ and $k\in G$ is chosen uniformly at random. Then, for any probabilistic adversary $\mathcal{A}$, the success probability of solving the EFP is negligible, specifically, bounded by
\begin{align*}
O\left( \frac{st}{|G|} \right),
\end{align*}
where $s$ and $t$ are the amount of encryption/decryption- and permutation/inverse permutation-oracle queries, respectively.
\end{thm}

By a basic reduction, and for the latter, by an inference result, we also get that
\begin{thm}
\textbf{[Informal]} If $P$ is a super pseudorandom permutation on $G$ and $k\in G$ is chosen uniformly at random. For any probabilistic adversary $\mathcal{A}$, the success probability of solving the EFP is negligible.
\end{thm}

\begin{cor}
\textbf{[Informal]} If $P$ is a super pseudorandom permutation on $G$ and $k\in G$ is chosen uniformly at random. For any \underline{polynomial-time} probabilistic adversary $\mathcal{A}$, the success probability of solving the CP is negligible.
\end{cor}

With the same bound as in Theorem~\ref{IntuitiveEFP}, we find that

\begin{thm}
\textbf{[Informal]} For any probabilistic adversary $\mathcal{A}$, limited to polynomially many encryption- and decryption-oracle queries and polynomially many permutation- and inverse permutation-oracle queries, the Group EM scheme over a group $G$ is a super pseudorandom permutation.
\end{thm}

We then apply a Slide Attack, to find an attack which matches the bound given above.

Considering the \textbf{Group Feistel cipher}, whose encryption algorithm consists of multiple uses of the round function
\begin{align*}
\mathcal{F}_f(x,y) = (y,x\cdot f(y)),
\end{align*}
where $f$ is a pseudorandom function on $G$, we show that the \textbf{$3$-round Feistel cipher} is pseudorandom but is not super pseudorandom, regardless of the underlying group $G$. We then note that the $4$-round Feistel cipher is super pseudorandom as proven in \citep{PatelRamzanSundaram}.

Finally, we consider the \textbf{Group Even-Mansour scheme instantiated using a $4$-round Feistel cipher} over $G^2 = G\times G$, which uses the encryption algorithm
\begin{align*}
\Psi_k^{f,g}(m) = \mathcal{F}_{g,f,f,g}(m \cdot k)\cdot k,
\end{align*}
where $f$ and $g$ are modelled as random functions, $m\in G^2$ the plaintext, and $k \in G^2$ is a uniformly random key. We then show one of our main results:

\begin{thm}
\textbf{[Informal]} For any probabilistic $4$-oracle adversary $\mathcal{A}$ with at most
\begin{itemize}
\item $q_c$ queries to the $\Psi$- and inverse $\Psi$-oracles (or random oracles),
\item $q_f$ queries to the $f$-oracle, and
\item $q_g$ queries to the $g$-oracle,
\end{itemize}
we have that the success probability of $\mathcal{A}$ distinguishing between $\Psi$ and a random oracle, is bounded by
\begin{align*}
 (2q_c^2 +4q_fq_c + 4q_gq_c + 2q_c^2 - 2q_c)|G|^{-1} + 2\cdot \begin{pmatrix}
	q_c \\ 2
	\end{pmatrix}(2|G|^{-1} + |G|^{-2}).
\end{align*}
\end{thm}

We may also rewrite our main theorem as the following:
\begin{thm}
\textbf{[Informal]} For any $4$-oracle adversary $\mathcal{A}$, with at most $q$ total queries, we have that the success probability of $\mathcal{A}$ distinguishing between $\Psi$ and a random oracle, is bounded by
\begin{align*}
2(3q^2-2q)|G|^{-1} + (q^2-q)|G|^{-2}.
\end{align*}
\end{thm}

We note that this main result is due to \citep{GentryRamzan}, however, we consider a one-key group version and add details to their proof sketches.

\subsection{Outline of Paper}
In Section~\ref{GenDefns}, we state the assumptions for this paper. In Section~\ref{GenDefns}, we give definitions that hold for the paper in general, leaving specialized definitions to the various sections. In Section~\ref{Even-MansourSection}, we introduce the generalized EM scheme over arbitrary groups, stating and proving some results about it. In Section~\ref{FeistelSection}, we define the generalized Feistel cipher over arbitrary groups and prove a few small results about it. In Section~\ref{ImplementEM}, we consider an implementation of the generalized EM scheme using the generalized Feistel cipher as the public permutation. In Section~\ref{ConclusionSection}, we give our concluding remarks.

\section{General Definitions}\label{GenDefns}
In the following, we work in the Random Oracle Model such that we may assume the existence of a random permutation oracle on group elements. We let $\mathcal{G}$ be the family of all finite groups, e.g. a group $G\in\mathcal{G}$ is a pair of the set $G$ and operation $\cdot$ satisfying the group axioms. We also assume that for any group $G\in \mathcal{G}$, $|G|\leq 2^{poly(n)}$ for some $n\in \mathbb{N}$ and some polynomial $poly(\cdot)$.

We will need the concept of pseudorandom, which is also called indistinguishable from random, in several forms. On notation, we write $x\in_R X$ for an element chosen uniformly at random from a set $X$. In the following, we consider the positive integer $\lambda$ to be the security parameter, specified in unary per convention. We assume that for each $\lambda$ there exists a uniquely specified group $G(\lambda) = G_\lambda \in \mathcal{G}$ with size $|G_\lambda| \geq 2^\lambda$.

\begin{defn}
Let $F_{m,n}:G_\lambda \times G_m \rightarrow G_n$, for $G_m,G_n \in \mathcal{G}$, be an efficient, keyed function. $F_{m,n}$ is a \textbf{pseudorandom function (PRF)} if for all probabilistic distinguishers $\mathcal{A}$, limited to only polynomially many queries to the function-oracle, there exists a negligible function $negl(\cdot)$, such that
\begin{align*}
\left| \underset{k \in_R G_\lambda}{Pr}\left[ \mathcal{A}^{F_{m,n}(k,\cdot)}(\lambda)=1 \right] - \underset{\pi \in_R \mathfrak{F}_{G_m\rightarrow G_n}}{Pr}\left[ \mathcal{A}^{\pi(\cdot)}(\lambda)=1 \right] \right| \leq negl(\lambda),
\end{align*}
where $\mathfrak{F}_{G_m\rightarrow G_n}$ is the set of functions from $G_m$ to $G_n$.
\end{defn}

If $F:G \times G \rightarrow G$ is a pseudorandom function, we say that it is a \textbf{pseudorandom function on $G$}.

\begin{defn}
Let $P:G_\lambda \times G \rightarrow G$ be an efficient, keyed permutation. $P$ is a \textbf{pseudorandom permutation (PRP)} if for all probabilistic distinguishers $\mathcal{A}$, limited to only polynomially many queries to the permutation-oracle, there exists a negligible function $negl(\cdot)$, such that
\begin{align*}
\left| \underset{k \in_R G_\lambda}{Pr}\left[ \mathcal{A}^{P(k,\cdot)}(\lambda)=1 \right] - \underset{\pi \in_R \mathfrak{P}_{G\rightarrow G}}{Pr}\left[ \mathcal{A}^{\pi(\cdot)}(\lambda)=1 \right] \right| \leq negl(\lambda),
\end{align*}
where $\mathfrak{P}_{G\rightarrow G}$ is the set of permutations on $G$.
\end{defn}

\begin{defn}
Let $P:G_\lambda \times G \rightarrow G$ be an efficient, keyed permutation. $P$ is said to be a \textbf{super pseudorandom permutation (SPRP)} if for all probabilistic distinguishers $\mathcal{A}$, limited to only polynomially many queries to the permutation- and inverse permutation-oracles, there exists a negligible function $negl(\cdot)$, such that
\begin{align*}
\left| \underset{k \in_R G_\lambda}{Pr}\left[ \mathcal{A}^{P(k,\cdot),P^{-1}(k,\cdot)}(\lambda)=1 \right] - \underset{\pi \in_R \mathfrak{P}_{G\rightarrow G}}{Pr}\left[ \mathcal{A}^{\pi(\cdot),\pi^{-1}(\cdot)}(\lambda)=1 \right] \right| \leq negl(\lambda),
\end{align*}
where $\mathfrak{P}_{G\rightarrow G}$ is the set of permutations on $G$.
\end{defn}

A (super) pseudorandom permutation $P:G \times G \rightarrow G$ is said to be a \textbf{(super) pseudorandom permutation on $G$}.

\newpage
\section{Even-Mansour}\label{Even-MansourSection}
We first remark that the results in this section were initially proven in a project prior to the start of the thesis but were further worked on to complement this thesis. Thus we have chosen to include parts of it, while this inclusion accounts for the brevity in certain results. We begin by defining the one-key Even-Mansour scheme over arbitrary groups, which we will refer to as the Group EM scheme.

\begin{defn}
We define the \textbf{Group Even-Mansour scheme} to be the triple of a key generation algorithm, encryption algorithm, and decryption algorithm. The key generation algorithm takes as input the security parameter $1^\lambda$, fixes and outputs a group $G\in_R \mathcal{G}$ with $|G|\geq 2^\lambda$, and outputs a key $k\in_R G$. The encryption algorithm $E_k(m)$ takes as input the key $k$ and a plaintext $m\in G$ and outputs
\begin{align*}
E_k(m) = P(m\cdot k) \cdot k,
\end{align*}
where $P$ is the public permutation. The decryption algorithm $D_k(c)$ takes as input the key $k$ and a ciphertext $c\in G$ and outputs
\begin{align*}
D_k(c) = P^{-1}(c \cdot k^{-1}) \cdot k^{-1},
\end{align*}
where $P^{-1}$ is the inverse public permutation. This definition satisfies correctness.
\end{defn}

\subsection{Two Forms of Security for the Group EM Scheme}
In this subsection, we prove classical results about our new scheme. We do so by considering Even and Mansour's two notions of security: the Existential Forgery Problem and the Cracking Problem, the Cracking Problem being the stronger of the two.

\begin{defn}
In the \textbf{Existential Forgery Problem} (EFP), we consider the following game:
\begin{enumerate}
\item A group $G\in \mathcal{G}$ and a key $k\in_R G$ are generated.
\item The adversary $\mathcal{A}$ gets the security parameter, in unary, and the group $G$.
\item $\mathcal{A}$ receives oracle access to the $E_k, D_k, P,$ and $P^{-1}$ oracles.
\item $\mathcal{A}$ eventually outputs a pair $(m,c)$.
\end{enumerate}
If $E_k(m)=c$, and $(m,c)$ has not been queried before, we say that $\mathcal{A}$ succeeds.

In the \textbf{Cracking Problem} (CP), we consider the following game:
\begin{enumerate}
\item A group $G\in \mathcal{G}$ and a key $k\in_R G$ are generated.
\item The adversary $\mathcal{A}$ gets the security parameter, in unary, and the group $G$.
\item $\mathcal{A}$ is presented with $E_k(m_0)=c_0\in_R G$.
\item $\mathcal{A}$ receives oracle access to the $E_k, D_k, P,$ and $P^{-1}$ oracles, but the decryption oracle outputs $\perp$ if $\mathcal{A}$ queries $c=c_0$.
\item $\mathcal{A}$ outputs a plaintext $m$.
\end{enumerate}
If $D_k(c_0)=m$, then we say that $\mathcal{A}$ succeeds. The \textbf{success probability} is the probability that on a uniformly random chosen encryption $c_0 = E_k(m_0)$, $\mathcal{A}$ outputs $m_0$.
\end{defn}

Even and Mansour show that polynomial-time EFP security infers poly\-nomial-time CP security. There are no limiting factors prohibiting the problems and inference result from being employed on groups. In fact, there is nothing disallowing the use of the same proof of the EFP security for the EFP security of the one-key EM scheme, as noted in \citep{DKS}, which we therefore omit. Indeed, by redefining notions in the \citep{EM} proof to take into account that we are working over a not necessarily abelian group, we are able to prove that the Group EM scheme satisfies the EFP notion of security, specifically the following.

\begin{thm}\label{MainEFP}
Assume $P\in_R\mathfrak{P}_{G\rightarrow G}$ and let the key $k\in_R G$. For any probabilistic adversary $\mathcal{A}$, the success probability of solving the EFP is bounded by
\begin{align*}
Succ(\mathcal{A}) = Pr_{k,P}\left[ EFP(\mathcal{A})=1\right] = O\left( \frac{st}{|G|} \right),
\end{align*}
where $s$ is the number of $E/D$-queries and $t$ is the number of $P/P^{-1}$-queries, i.e. the success probability is negligible.
\end{thm}

By the Even and Mansour inference result, we get the corollary below.

\begin{cor}
Assume $P\in_R\mathfrak{P}_{G\rightarrow G}$ and let the key $k\in_R G$. For any probabilistic polynomial-time (PPT) adversary $\mathcal{A}$, the success probability of solving the Cracking Problem is negligible.
\end{cor}

As Even and Mansour also note, the above results may be extended to instances where the permutation is a pseudorandom permutation by a simple reduction. Hence, we get the following two results.

\begin{thm}
Assume $P$ is a pseudorandom permutation on $G\in \mathcal{G}$ and let the key $k\in_R G$. For any probabilistic adversary $\mathcal{A}$ with only polynomially many queries to its oracles, the success probability of solving the Existential Forgery Problem is negligible.
\end{thm}

\begin{cor}
Assume $P$ is a pseudorandom permutation on $G\in \mathcal{G}$ and let the key $k\in_R G$. For any probabilistic polynomial-time (PPT) adversary $\mathcal{A}$, the success probability of solving the Cracking Problem is negligible.
\end{cor}

\subsection{Pseudorandomness Property of the Group EM Scheme}
Although the above notions of security are strong, we are more interested in any pseudorandomness property the Group EM scheme offers us. Kilian and Rogaway \citep{Kilr} show that the one-key EM scheme satisfies the pseudorandom permutation property, i.e. with only an encryption oracle and the permutation oracles, the EM scheme is indistinguishable from random to any adversary with only polynomially many queries to its oracles. We note that they only show the pseudorandomness property, but state in their discussion section that their proof may be adapted to include a decryption oracle, i.e. that the one-key EM scheme satisfies the super pseudorandom permutation property. Having done the analysis with the decryption oracle, over an arbitrary group, we concur. However, we were also able to generalize the \citep{Kilr} proof to a one-key construction. This not entirely remarkable as the key $k$ will usually be different from its group inverse, hence we were able to use the same proof, but with adjustments to the games and their analysis. The proof is given in the appendix for posterity. For completeness, we present the result as the following theorem.

\begin{thm}
Assume $P\in_R\mathfrak{P}_{G\rightarrow G}$ and let the key $k\in_R G$. For any probabilistic adversary $\mathcal{A}$, limited to polynomially many $E/D$- and $P/P^{-1}$-oracle queries, the adversarial advantage of $\mathcal{A}$ is bounded by
\begin{align*}
\text{Adv}(\mathcal{A}) \defeq \left| Pr\left[ \mathcal{A}_{E_k,D_k}^{P,P^{-1}} = 1 \right] - Pr\left[\mathcal{A}_{\pi,\pi^{-1}}^{P,P^{-1}} = 1 \right]\right| = \mathcal{O}\left(\frac{st}{|G|}\right).
\end{align*}
where $s$ is the number of $E/D$-queries and $t$ is the number of $P/P^{-1}$-queries, i.e. the success probability is negligible.
\end{thm}

Stated simply,

\begin{thm}
For any probabilistic adversary $\mathcal{A}$, limited to polynomially many $E/D$- and $P/P^{-1}$-oracle queries, the Group EM scheme over a group $G$ is a super pseudorandom permutation.
\end{thm}

By removing the decryption oracle, we get the following corollary:

\begin{cor}
For any probabilistic adversary $\mathcal{A}$, limited to polynomially many $E$- and $P/P^{-1}$-oracle queries, the Group EM scheme over a group $G$ is a pseudorandom permutation.
\end{cor}

\begin{rem}
We see that in the group $((\mathbb{Z}/2\mathbb{Z})^n,\oplus)$, our Group EM scheme reduces to the one-key EM scheme given in \citep{DKS}. The proof given in \citep{DKS} proves the security of the scheme, and the proof given in \citep{Kilr} proves the pseudorandomness, equivalently to our claims.
\end{rem}

It can be proven that a multiple round Group EM scheme is an SPRP because the security only depends on the last round, which is also an SPRP.

\subsection{Slide Attack}
We would like to show that the security bound that we have found above is optimal, so we slightly alter the simple optimal attack on the Single-Key Even-Mansour cipher as constructed in \citep{DKS}. The original version works for abelian groups with few adjustments and \citep{DKS} also present another slide attack against a modular addition DESX construction.

Consider the one-key Group Even-Mansour cipher
\begin{align*}
E(x) = P(x \cdot k) \cdot k,
\end{align*}
over a group $G$ with binary operation $\cdot$, where $P$ is a publicly available permutation oracle, $x\in G$, and $k\in_R G$. Define the following values:
\begin{align*}
x=x, \hspace*{5pt} y=x \cdot k, \hspace*{5pt} z= P(y), \hspace*{5pt} w=E(x)=P(x \cdot k) \cdot k.
\end{align*}
We hereby have that $w\cdot y^{-1} = z \cdot x^{-1} $. Consider the attack which follows.
\begin{enumerate}
\item For $d = \sqrt{|G|}$ arbitrary values $x_i\in G$, $i=1,\ldots, d$, and $d$ arbitrary values $y_i\in G$, $i=1,\ldots, d$, query the $E$-oracle on the $x_i$'s and the $P$-oracle on the $y_i$'s. Store the values in a hash table as
\begin{align*}
(E(x_i)\cdot y_i^{-1}, P(y_i)\cdot x_i^{-1}, i),
\end{align*}
sorted by the first coordinate.
\item If there exists a match in the above step, i.e. $E(x_i)\cdot y_i^{-1} = P(y_i)\cdot x_i^{-1}$ for some $i$, check the guess that $k = x_i^{-1}\cdot y_i$.
\end{enumerate}
It can be seen by the Birthday Problem\footnote{Considering the approximation $p(n)\approx \tfrac{n^2}{2m}$, where $p(n)$ is the probability of there being a Birthday Problem collision from $n$ randomly chosen elements from the set of $m$ elements, then $p(\sqrt{|G|})\approx \tfrac{\sqrt{|G|}^2}{2|G|}=1/2$.}, that with non-negligible probability, there must exist a slid pair $(x_i,y_i)$ satisfying the above property, i.e. there exists $1\leq i \leq d$ such that $k = x_i^{-1} \cdot y_i$. For a random pair $(x,y)\in G^2$ it holds that $E(x) = P(y) \cdot x^{-1} \cdot y$ with probability $|G|^{-1}$, so we expect few, if any, collisions in the hash table, including the collision by the slid pair where the correct key $k$ is found. The data complexity of the attack is $d$ $E$-oracle queries and $d$ $P$-oracle queries. Hence the attack bound $d^2 = |G|$, which matches the lower bound given in Theorem~\ref{MainEFP} and Theorem~\ref{PseudoEMbounded}. We have therefore found that our scheme is optimal.

\newpage
\section{Feistel}\label{FeistelSection}
We now consider the Feistel cipher over arbitrary groups, which we will call the Group Feistel cipher. The following is a complement to \citep{PatelRamzanSundaram} who treat the Group Feistel cipher construction with great detail. Our main accomplishment in this section is the settling of an open problem posed by them.

\subsection{Definitions}
We define a Feistel cipher over a group $(G,\cdot)$ as a series of round functions on elements of $G\times G=G^2$.

\begin{defn}
Given an efficiently computable but not necessarily invertible function $f: G \rightarrow G$, called a \textbf{round function}, we define the \textbf{1-round Group Feistel cipher} $\mathcal{F}_{f}$ to be
\begin{align*}
\mathcal{F}_{f}: 	G \times G &\longrightarrow G \times G,\\
				(x,y) &\longmapsto (y, x \cdot f(y)).
\end{align*}
In the case where we have multiple rounds, we index the round functions as $f_i$, and denote the \textbf{$r$-round Group Feistel cipher} by $\mathcal{F}_{f_1,\ldots,f_r}$. We concurrently denote the input to the $i$'th round by $(L_{i-1},R_{i-1})$ and having the output $(L_i,R_i) = (R_{i-1}, L_{i-1}\cdot f_i(R_{i-1}))$, where $L_i$ and $R_i$ respectively denote the left and right parts of the $i$'th output.
\end{defn}

Note that if $(L_i,R_i)$ is the $i$'th round output, we may invert the $i$'th round by setting $R_{i-1}:=L_i$ and then computing $L_{i-1}:= R_i \cdot (f_i(R_{i-1}))^{-1}$ to get $(L_{i-1},R_{i-1})$. As this holds for all rounds, regardless of the invertibility of the round functions, we get that an $r$-round Feistel cipher is invertible for all $r$.

Let $F:G_\lambda \times G \rightarrow G$ be a pseudorandom function. We define the keyed permutation $F^{(r)}$ as
\begin{align*}
F^{(r)}_{k_1,\ldots,k_r}(x,y) \defeq \mathcal{F}_{F_{k_1},\ldots, F_{k_r}}(x,y).
\end{align*}
We sometimes index the keys as $1,2,\ldots, r$, or omit the key index entirely.

\subsection{Results}
For completeness, we show some of the preliminary results for Group Feistel ciphers, not considered in \citep{PatelRamzanSundaram}.

We first note that $F^{(1)}$ is \textit{not} a pseudorandom permutation as 
\begin{align*}
F^{(1)}_{k_1}(L_0,R_0) = (L_1,R_1) = (R_0,L_0\cdot F_{k_1}(R_0)),
\end{align*}
such that any distinguisher $\mathcal{A}$ need only compare $R_0$ to $L_1$.

Also $F^{(2)}$ is \textit{not} a pseudorandom permutation: Consider a pseudorandom function $F$ on $G$. Pick $k_1,k_2\in_R G_\lambda$. Distinguisher $\mathcal{A}$ sets $(L_0,R_0)=(1,g)$ for some $g\in G$, where $1$ is the identity element of $G$, then queries $(L_0,R_0)$ to its oracle and receives,
\begin{center}
$L_2 = L_0 \cdot F_{k_1}(R_0) = F_{k_1}(g)$ and $R_2 = R_0 \cdot F_{k_2}(L_0\cdot F_{k_1}(R_0)) = g \cdot F_{k_2}(F_{k_1}(g))$.
\end{center}
On its second query, the distinguisher $\mathcal{A}$ lets $L_0 \in G \setminus \lbrace 1\rbrace$ but $R_0=g$, such that it receives
\begin{center}
$L_2 = L_0 \cdot F_{k_1}(R_0) = L_0 \cdot F_{k_1}(g)$ and $R_2 = g \cdot F_{k_2}(L_0 \cdot F_{k_1}(g))$.
\end{center}
As $\mathcal{A}$ may find the inverse to elements in $G$, $\mathcal{A}$ acquires $(F_{k_1}(g))^{-1}$, and by so doing, may compute $L_2 \cdot (F_{k_1}(g))^{-1} = L_0$. If $F^{(2)}$ were random, this would only occur negligibly many times, while $\mathcal{A}$ may query its permutation-oracle polynomially many times such that if $L_0$ is retrieved non-negligibly many times out of the queries, $\mathcal{A}$ is able to distinguish between a random permutation and $F^{(2)}$ with non-negligible probability.

As one would expect, the $3$-round Group Feistel cipher (see Figure~\ref{3roundFeistel}) is indeed a pseudorandom permutation.

\begin{figure}
\centering
\begin{subfigure}[b]{0.49\textwidth}
\centering
\begin{tikzpicture}

    \tikzstyle{dot} = [
		fill,
		shape=circle,
		minimum size=4pt,
		inner sep=0pt,
	]

    \foreach \z in {1, 2,...,3} {
        \node[draw,thick,minimum width=1cm] (g\z) at ($\z*(0,-1.5cm)$)  {$f_\z$};
        \node (prik\z) [dot, left of = g\z, node distance = 2cm, scale=0.8] {};
        \draw[thick,-latex] (g\z) -- (prik\z);
    
    }
    
    \foreach \z in {1, 2} {
   	 	\draw[thick,latex-latex] (g\z.east) -| +(1.5cm,-0.5cm) -- ($(prik\z) - (0,1cm)$) -- ($(prik\z.north) - (0,1.5cm)$);
   	 	\draw[thick] (prik\z.south) -- ($(prik\z)+(0,-0.5cm)$) -- ($(g\z.east) + (1.5cm,-1cm)$) -- +(0,-0.5cm);
    }

    \node (p0) [above of = g1, minimum width=5cm,minimum height=0.5cm,node distance=1cm] {}; 
    \node (l0) [above of = prik1,node distance=1cm] {$L_0$};
    \node (r0) [right of = l0, node distance = 4cm] {$R_0$};
    \draw[thick,-latex] (l0 |- p0.south) -- (prik1.north);
    \draw[thick] ($(g1.east)+(1.5cm,0)$) -- +(0,0.75cm);

    \node (p3) [below of = g3, minimum width=5cm,minimum height=0.5cm,node distance=1.75cm] {}; 
    \node (l3) [below of = prik3,node distance=1.75cm] {$L_3$};
    \node (r3) [right of = l3, node distance = 4cm] {$R_3$};
    \draw[thick,latex-latex] (g3.east) -| +(1.5cm,-0.5cm) -- ($(prik3) - (0,1cm)$) -- (prik3 |- p3.north);
    \draw[thick,-latex] (prik3.south) -- ($(prik3)+(0,-0.5cm)$) -- ($(g3.east) + (1.5cm,-1cm)$) -- +(0,-0.5cm);

\end{tikzpicture}
\captionof{figure}[$3$-round Group Feistel cipher.]{$3$-round Group Feistel cipher.}
\label{3roundFeistel}
\end{subfigure}
\begin{subfigure}[b]{0.49\textwidth}
\centering
\begin{tikzpicture}

    \tikzstyle{dot} = [
		fill,
		shape=circle,
		minimum size=4pt,
		inner sep=0pt,
	]

    
    \node[draw,thick,minimum width=1cm] (f1) at ($1*(0,-1.5cm)$)  {$g$};
    \node (xor1) [dot, left of = f1, node distance = 2cm, scale=0.8] {};
    \draw[thick,-latex] (f1) -- (xor1);   
    \node[draw,thick,minimum width=1cm] (f2) at ($2*(0,-1.5cm)$)  {$f$};
    \node (xor2) [dot, left of = f2, node distance = 2cm, scale=0.8] {};
    \draw[thick,-latex] (f2) -- (xor2);   
    \node[draw,thick,minimum width=1cm] (f3) at ($3*(0,-1.5cm)$)  {$f$};
    \node (xor3) [dot, left of = f3, node distance = 2cm, scale=0.8] {};
    \draw[thick,-latex] (f3) -- (xor3);   
    \node[draw,thick,minimum width=1cm] (f4) at ($4*(0,-1.5cm)$)  {$g$};
    \node (xor4) [dot, left of = f4, node distance = 2cm, scale=0.8] {};
    \draw[thick,-latex] (f4) -- (xor4);

    \foreach \z in {1, 2,...,3} {
   	 	\draw[thick,latex-latex] (f\z.east) -| +(1.5cm,-0.5cm) -- ($(xor\z) - (0,1cm)$) -- ($(xor\z.north) - (0,1.5cm)$);
   	 	\draw[thick] (xor\z.south) -- ($(xor\z)+(0,-0.5cm)$) -- ($(f\z.east) + (1.5cm,-1cm)$) -- +(0,-0.5cm);
    }

    \node (p0) [above of = f1, minimum width=5cm,minimum height=0.5cm,node distance=1cm] {}; 
    \node (l0) [above of = xor1,node distance=1cm] {$x^L$};
    \node (mid0) [above of = xor1,node distance = .5cm] {};
    \node (kl0) [right of = mid0,node distance=.3cm] {$\cdot k^L$};
    \node (r0) [right of = l0, node distance = 4cm] {$x^R$};
    \node (kr0) [right of = kl0,node distance= 4.05cm] {$\cdot k^R$};
    \draw[thick,-latex] (l0 |- p0.south) -- (xor1.north);
    \draw[thick] ($(f1.east)+(1.5cm,0)$) -- +(0,0.75cm);

    \node (p4) [below of = f4, minimum width=5cm,minimum height=0.5cm,node distance=1.75cm] {}; 
    \node (l4) [below of = xor4,node distance=1.75cm] {$y^L$};
    \node (kl4) [below of = kl0,node distance=6.15cm] {$\cdot k^L$};
    \node (r4) [right of = l4, node distance = 4cm] {$y^R$};
    \node (kr4) [right of = kl4,node distance= 4.05cm] {$\cdot k^R$};
    \draw[thick,latex-latex] (f4.east) -| +(1.5cm,-0.5cm) -- ($(xor4) - (0,1cm)$) -- (xor4 |- p4.north);
    \draw[thick,-latex] (xor4.south) -- ($(xor4)+(0,-0.5cm)$) -- ($(f4.east) + (1.5cm,-1cm)$) -- +(0,-0.5cm);

\end{tikzpicture}
\captionof{figure}[Group EM scheme with Feistel.]{Group EM scheme with Feistel.\footnotemark}
\label{GendGenRamPic}
\end{subfigure}
\captionof{figure}{Encryption schemes.}
\end{figure}
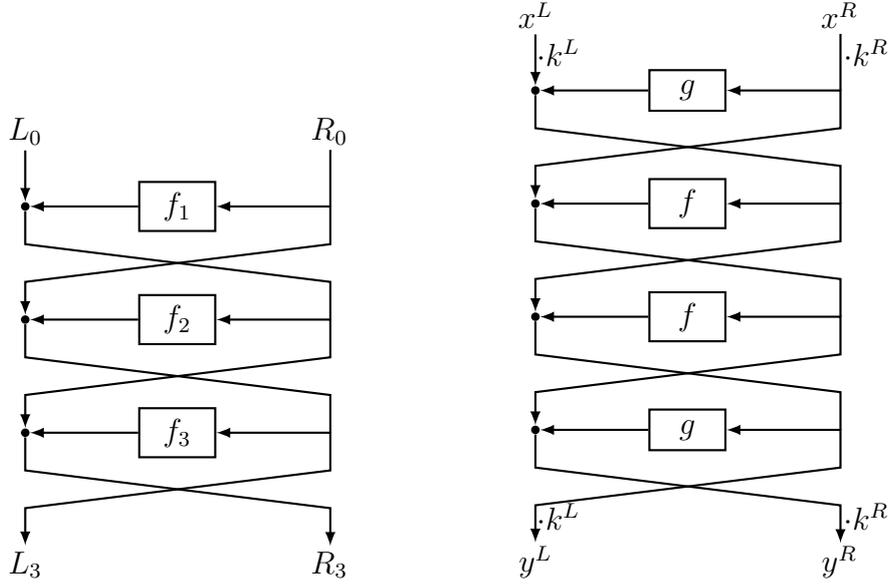

\footnotetext{TikZ figure adapted from \citep{TikZhelp}.}

\begin{thm}
If $F$ is a pseudorandom function on $G$, then $F^{(3)}$ is a pseudorandom permutation on $G$.
\end{thm}

The proof of this proposition can be generalized from the proof given in Katz and Lindell \citep{KL} of the analogous result over bit-strings with XOR, with no difficulties. We therefore omit it here.

Among the considerations in \citep{PatelRamzanSundaram}, they showed that the $3$-round Feistel cipher over abelian groups was not super pseudorandom, but left as an open problem a proof over non-abelian groups. We present such a proof now.

\begin{prop}
The $3$-round Group Feistel cipher is not super pseudorandom.
\end{prop}

\begin{proof}
The proof is a counter-example using the following procedure:
\begin{enumerate}
\item Choose two oracle-query pairs in $G\times G$: $(L_0,R_0)$ and $(L'_0,R_0)$ where $L_0\neq L'_0$.
\item Query the encryption oracle to get $(L_3,R_3)$ and $(L'_3,R'_3)$.
\item Query $(L''_3,R''_3)= (L'_3,L_0\cdot (L'_0)^{-1} \cdot R'_3)$ to the decryption oracle.
\item If $R''_0=L'_3\cdot (L_3)^{-1} \cdot R_0$, guess that the oracle is $F^{(3)}$, else guess random.
\end{enumerate}
For $F^{(3)}$, this algorithm succeeds with probability $1$. For a random permutation, this algorithm succeeds negligibly often.
\end{proof}

For super pseudorandomness of the $4$-round Group Feistel cipher, we refer the reader to \citep{PatelRamzanSundaram}. In the paper, they show a strong result using certain hash functions as round functions, from which the following is a corollary.

\begin{cor}
Let $G$ be a group, with characteristic other than $2$, and let $f,g: G_\lambda \times G \rightarrow G$ be pseudorandom functions. Then, for any adversary $\mathcal{A}$ with polynomially many queries to its $E/D$-oracles, the family $\mathcal{P}$ of permutations on $G\times G$ consisting of permutations of the form $F^{(4)}=\mathcal{F}_{g,f,f,g}$ are indistinguishable from random, i.e. super pseudorandom permutations (SPRPs).
\end{cor}

\newpage
\section{Implementing the Group Even-Mansour Scheme}\label{ImplementEM}
Now that we have shown that both the Even-Mansour scheme and the Feistel cipher are generalizable to arbitrary groups, we might consider how to implement one given the other. Gentry and Ramzan \citep{GentryRamzan} considered exactly this for the two-key EM scheme over $(\mathbb{Z}/2\mathbb{Z})^n$. However, their paper only had sketches of proofs and refer to another edition of the paper for full details. As we are unable to find a copy in the place that they specify it to exist, and as we generalize their result non-trivially, we have decided to fill in the details while generalizing their proof.

In this section, we consider a generalized version of the Gentry and Ramzan \citep{GentryRamzan} construction, namely, the Group Even-Mansour scheme on $G^2$ instantiated with a $4$-round Group Feistel cipher as the public permutation:
\begin{align*}
\Psi_{k}^{f,g}(x)=\mathcal{F}_{g,f,f,g}(x\cdot k)\cdot k,
\end{align*}
where $k=(k^L,k^R)\in G^2$ is a key consisting of two subkeys, chosen independently and uniformly at random, and $f$ and $g$ are round functions on $G$, modelled as random function oracles, available to all parties, including the adversary. We consider the operation $x\cdot k$ for $x=(x^L,x^R)\in G^2$, to be the coordinate-wise group operation, but do not otherwise discern between it and the group operation $\cdot$ on elements of $G$. In the following, we shall follow the proof in \citep{GentryRamzan} closely. However, we make quite a few modifications, mostly due to the nature of our generalization. Note that we consider a one-key scheme, as opposed to the two-key version in \citep{GentryRamzan} (see Figure~\ref{GendGenRamPic}.) Our main theorem for this section is the following.

\begin{thm}\label{GentryRamzanMain}
Let $f,g$ be modelled as random oracles and let the subkeys of $k=(k^L,k^R)\in G^2$ be chosen independently and uniformly at random. Let $\Psi_k^{f,g}(x)=\mathcal{F}_{g,f,f,g}(x\cdot k)\cdot k$, and let $R\in_R \mathfrak{P}_{G^2 \rightarrow G^2}$. Then, for any probabilistic $4$-oracle adversary $\mathcal{A}$ with at most
\begin{itemize}
\item $q_c$ queries to $\Psi$ and $\Psi^{-1}$ (or $R$ and $R^{-1}$),
\item $q_f$ queries to $f$, and
\item $q_g$ queries to $g$,
\end{itemize}
we have
\begin{align*}
&\left| Pr\left[ \mathcal{A}^{\Psi,\Psi^{-1},f,g} = 1\right] - Pr\left[ \mathcal{A}^{R,R^{-1},f,g} = 1\right] \right| \\
	&\hspace*{20pt}\leq (2q_c^2 +4q_fq_c + 4q_gq_c + 2q_c^2 - 2q_c)|G|^{-1} + 2\cdot \begin{pmatrix}
	q_c \\ 2
	\end{pmatrix}(2|G|^{-1} + |G|^{-2}).
\end{align*}
\end{thm}

\subsection{Definitions}
Before we can begin the proof, we will need several definitions all of which are identical to the \citep{GentryRamzan} definitions, up to rewording.

\begin{defn}
Let $P$ denote the permutation oracle (either $\Psi$ or $R$), $\mathcal{O}^f$ and $\mathcal{O}^g$ the $f$ and $g$ oracles, respectively. We get the transcripts: $T_P$, the set of all $P$ queries, $T_f$, the set of all $f$ queries, and $T_g$, the set of all $g$ queries, i.e. the sets
\begin{align*}
T_P &= \lbrace \langle x_1,y_1 \rangle, \langle x_2,y_2 \rangle, \cdots, \langle x_{q_c},y_{q_c} \rangle \rbrace_P, \\
T_f &= \lbrace \langle x'_1,y'_1 \rangle, \langle x'_2,y'_2 \rangle, \cdots, \langle x'_{q_f},y'_{q_f} \rangle \rbrace_f, \\
T_g &= \lbrace \langle x''_1,y''_1 \rangle, \langle x''_2,y''_2 \rangle, \cdots, \langle x''_{q_g},y''_{q_g} \rangle \rbrace_g.
\end{align*}
We discern between two types of oracle queries: Cipher queries $(+,x)=P(x)$ and $(-,y)=P^{-1}(y)$; Oracle queries $(\mathcal{O}^f,x')$ and $(\mathcal{O}^g,x'')$, respectively $f$- and $g$-oracle queries.
\end{defn}

As we have no bounds on the computational complexity of the adversary $\mathcal{A}$, we may assume that $\mathcal{A}$ is deterministic, as we did in the proof of Theorem~\ref{PseudoEMbounded}. Hence, we may consider an algorithm $C_\mathcal{A}$ which, given a set of $\mathcal{A}$'s queries, can determine $\mathcal{A}$'s next query.

\begin{defn}
For $0\leq i \leq q_c$, $0\leq j \leq q_f$, and $0\leq k \leq q_g$, the $i+j+k+1$'st query by $\mathcal{A}$ is
\begin{align*}
C_\mathcal{A}\left[ \lbrace \langle x_1,y_1 \rangle, \ldots, \langle x_{i},y_{i} \rangle \rbrace_P, \lbrace \langle x'_1,y'_1 \rangle, \ldots, \langle x'_{j},y'_{j} \rangle \rbrace_f, \lbrace \langle x''_1,y''_1 \rangle, \ldots, \langle x''_{k},y''_{k} \rangle \rbrace_g \right]
\end{align*}
where the upper equality case on the indexes is defined to be $\mathcal{A}$'s final output.
\end{defn}

\begin{defn}
Let $\sigma = (T_P,T_f,T_g)$ be a tuple of transcripts with length $q_c,q_f,q_g$, respectively. We say that $\sigma$ is a \textbf{possible $\mathcal{A}$-transcript} if for every $1\leq i \leq q_c, 1 \leq j \leq q_f$, and $1\leq k \leq q_g$,
\begin{align*}
C_\mathcal{A}\left[ \lbrace \langle x_1,y_1 \rangle, \ldots, \langle x_{i},y_{i} \rangle \rbrace_P, \lbrace \langle x'_1,y'_1 \rangle, \ldots, \langle x'_{j},y'_{j} \rangle \rbrace_f, \lbrace \langle x''_1,y''_1 \rangle, \ldots, \langle x''_{k},y''_{k} \rangle \rbrace_g \right] \\ \hspace*{70pt}\in \lbrace (+,x_{i+1}), (-,y_{i+1}), (\mathcal{O}^f, x'_{j+1}), (\mathcal{O}^g,x''_{k+1})\rbrace .
\end{align*}
\end{defn}

Let us define two useful ways in which we may answer $\mathcal{A}$'s queries other than what we have already defined.

\begin{defn}
Let $\tilde{\Psi}$ be the process where the $\Psi$- and $\Psi^{-1}$ cipher query oracles use $f$ and $g$, and $\mathcal{O}^f$ uses $f$, but $\mathcal{O}^g$ is replaced by $\mathcal{O}^h$ for another, independent, random function $h$.
\end{defn}

\begin{defn}
Let $\tilde{R}$ denote the process which answers all \underline{oracle} queries using $f$ and $g$, but answers the $i$'th \underline{cipher} query as follows.
\begin{enumerate}
\item If $\mathcal{A}$ queries $(+,x_i)$ and there exists $1 \leq j < i$, such that the $j$'th query-answer pair has $x_j=x_i$, return $y_i:=y_j$.
\item If $\mathcal{A}$ queries $(-,y_i)$ and there exists $1 \leq j < i$, such that the $j$'th query-answer pair has $y_j=y_i$, return $x_i:=x_j$.
\item Otherwise, return uniformly chosen element in $G^2$.
\end{enumerate}
\end{defn}

The latter definition may not be consistent with any function or permutation, so we formalize exactly this event.

\begin{defn}\label{inconsistent}
Let $T_P$ be a possible $\mathcal{A}$-cipher-transcript. $T_P$ is \textbf{inconsistent} if for some $1\leq i < j \leq q_c$ there exist cipher-pairs such that either
\begin{itemize}
\item $x_i = x_j$ but $y_i \neq y_j$, or
\item $x_i \neq x_j$ but $y_i=y_j$.
\end{itemize}
Any $\sigma$ containing such a transcript $T_P$ is called \textbf{inconsistent}.
\end{defn}

\begin{note}
Assume from now on that $\mathcal{A}$ never repeats any part of a query if the answer can be determined from previous queries, i.e. every possible $\mathcal{A}$-transcript $\sigma$ is consistent such that if $i\neq j$, then $x_i \neq x_j$, $y_i\neq y_j$, $x'_i \neq x'_j$, and $x''_i \neq x''_j$.
\end{note}

\begin{note}
Let $T_\Psi, T_{\tilde{\Psi}}, T_{\tilde{R}}, T_R$ denote the transcripts seen by $\mathcal{A}$ when its cipher queries are answered by $\Psi, \tilde{\Psi}, \tilde{R}, R$, respectively, and oracle queries by $\mathcal{O}^f$ and $\mathcal{O}^g$ (noting that in the case of $\tilde{\Psi}$, the function in the $\mathcal{O}^g$ has been replaced by another random function, $h$.) We also note that using this notation, we have that $\mathcal{A}^{\Psi,\Psi^{-1},f,g} = C_\mathcal{A}(T_\Psi)$ (and likewise for $\tilde{\Psi}, \tilde{R}$, and $R$.)
\end{note}

\subsection{Lemmas}
Now, let us begin finding results that will aid us in proving our main theorem. First, we will compare the distributions of $\tilde{R}$ and $R$, using a result by Naor-Reingold\footnote{The proof of the proposition follows the argument of Proposition 3.3 in \citep{NaorReingold}.}. Afterwards, we shall consider when the distributions of $\Psi$ and $\tilde{\Psi}$ are equal. Lastly, we shall consider when the distributions of $\tilde{\Psi}$ and $\tilde{R}$ are equal. Combining these results will allow us to prove our main theorem.

We remark that whenever we write $k= (k^L,k^R)\in_R G^2$, we mean that the subkeys are chosen independently and uniformly at random.

\begin{lem}\label{Lem37}
$\left| \underset{\tilde{R}}{Pr}\left[ C_\mathcal{A}(T_{\tilde{R}})=1 \right] - \underset{R}{Pr}\left[ C_\mathcal{A}(T_{R})=1 \right] \right| \leq \begin{pmatrix}
q_c \\ 2
\end{pmatrix}\cdot |G|^{-2}$.
\end{lem}

\begin{proof}
Let $\sigma$ be a possible and consistent $\mathcal{A}$-transcript, then
\begin{align*}
\underset{R}{Pr}\left[ T_R = \sigma \right] = \begin{pmatrix}
|G|^{2} \\ q_c \end{pmatrix} = \underset{\tilde{R}}{Pr} \left[ T_{\tilde{R}} = \sigma \mid T_{\tilde{R}} \textit{ is consistent} \right],
\end{align*}
simply because the only difference between $T_R$ and $T_{\tilde{R}}$ is in the cipher queries, and when $T_{\tilde{R}}$ is consistent, we have no overlap on the query-answer pairs, hence we need only consider how to choose $q_c$ elements from $|G|^2$ many possible elements, without replacement. Let us now consider the probability of $T_{\tilde{R}}$ being inconsistent. If $T_{\tilde{R}}$ is inconsistent for some $1\leq i < j \leq q_c$ then either $x_i=x_j$ and $y_i \neq y_j$, or $x_i \neq x_j$ and $y_i = y_j$. For any given $i,j$, this happens with at most probability $|G|^{-2}$, because if $x_i=x_j$ is queried, then the $\tilde{R}$-oracle would return the corresponding $y_i=y_j$, but if $x_i \neq x_j$ is queried, then the $\tilde{R}$-oracle would return a uniformly random element (and likewise if $y_i=y_j$ or $y_i\neq y_j$ were queried to the inverse $\tilde{R}$-oracle.) Hence,
\begin{align*}
\underset{\tilde{R}}{Pr} \left[ T_{\tilde{R}} \textit{ is inconsistent} \right] \leq \begin{pmatrix}
q_c \\ 2 \end{pmatrix} \cdot |G|^{-2}.
\end{align*}
We thereby get that,
\begin{footnotesize}
\begin{align*}
&\left| \underset{\tilde{R}}{Pr}\left[ C_\mathcal{A}(T_{\tilde{R}})=1 \right] - \underset{R}{Pr}\left[ C_\mathcal{A}(T_{R})=1 \right] \right| \\
&\leq \left| \underset{\tilde{R}}{Pr}\left[ C_\mathcal{A}(T_{\tilde{R}})=1 | T_{\tilde{R}} \textit{ is consistent} \right] - \underset{R}{Pr}\left[ C_\mathcal{A}(T_{R})=1 \right]\right| \cdot \underset{\tilde{R}}{Pr}\left[ T_{\tilde{R}} \textit{ is consistent} \right] \\
	&\hspace*{10pt} + \left| \underset{\tilde{R}}{Pr}\left[ C_\mathcal{A}(T_{\tilde{R}})=1 | T_{\tilde{R}} \textit{ is inconsistent} \right] - \underset{R}{Pr}\left[ C_\mathcal{A}(T_{R})=1 \right]\right| \cdot \underset{\tilde{R}}{Pr}\left[ T_{\tilde{R}} \textit{ is inconsistent} \right] \\
&\leq \underset{\tilde{R}}{Pr}\left[ T_{\tilde{R}} \textit{ is inconsistent} \right] \\
&\leq \begin{pmatrix} q_c \\ 2 \end{pmatrix} \cdot |G|^{-2},
\end{align*}
\end{footnotesize}
as the distribution over $R$ is independent of the (in)consistency of $T_{\tilde{R}}$.
\end{proof}

Let us now focus on the distributions of $T_\Psi$ and $T_{\tilde{\Psi}}$, to show that they are identical unless the input to $g$ in the cipher query to $\Psi$ is equal to the oracle input to $h$ in $\mathcal{O}^h$. In order to do so, we first define the event $\textsf{BadG}(k)$.

\begin{defn}
For every specific key $k=(k^L,k^R)\in_R G^2$, we define $\textsf{BadG}(k)$ to be the set of all possible and consistent $\mathcal{A}$-transcripts $\sigma$, satisfying at least one of the following:
\begin{itemize}
\item[\textbf{BG1}:] $\exists i,j, 1\leq i \leq q_c, 1 \leq j \leq q_g$, such that $x_i^R \cdot k^R = x''_j$, or
\item[\textbf{BG2}:] $\exists i,j, 1\leq i \leq q_c, 1 \leq j \leq q_g$, such that $y_i^L \cdot (k^L)^{-1} = x''_j$.
\end{itemize}
\end{defn}

\begin{lem}\label{Lem39}
Let $k=(k^L,k^R)\in_R G^2$. For any possible and consistent $\mathcal{A}$-transcript $\sigma=(T_P,T_f,T_g)$, we have
\begin{align*}
\underset{k}{Pr}\left[ \sigma \in \textsf{BadG}(k) \right] \leq \frac{2q_gq_c}{|G|}.
\end{align*}
\end{lem}

\begin{proof}
We know that $\sigma \in \textsf{BadG}(k)$ if one of \textbf{BG1} or \textbf{BG2} occur, hence, using the union bound,
\begin{align*}
\underset{k}{Pr}\left[ \sigma \in \textsf{BadG}(k) \right] &= \underset{k}{Pr}\left[ \textbf{BG1} \text{ occurs } \vee \textbf{BG2} \text{ occurs } | \sigma \right] \\
	&\leq  \underset{k}{Pr}\left[ \textbf{BG1} \text{ occurs } | \sigma \right] +  \underset{k}{Pr}\left[ \textbf{BG2} \text{ occurs } | \sigma \right] \\
	&\leq q_gq_c\cdot |G|^{-1} + q_gq_c\cdot |G|^{-1} \\
	&= 2q_gq_c\cdot |G|^{-1}.
\end{align*}
\end{proof}

\begin{lem}\label{NotBadPsitoBarPsi}
Let $\sigma$ be a possible and consistent $\mathcal{A}$-transcript, then
\begin{align*}
\underset{\Psi}{Pr}\left[ T_{\Psi} = \sigma | \sigma \not\in \textsf{BadG}(k) \right] = \underset{\tilde{\Psi}}{Pr}\left[ T_{\tilde{\Psi}} = \sigma \right].
\end{align*}
\end{lem}

\begin{proof}
We want to show that the query answers in the subtranscripts of the games $\Psi$ and $\tilde{\Psi}$ are equally distributed, under the condition that neither of the events \textbf{BG1} nor \textbf{BG2} occur in game $\Psi$. Fix the key $k=(k^L,k^R)\in_R G^2$. Recall that the adversary does not query an oracle if it can determine the answer from previous queries.

In both games, for any $\mathcal{O}^f$-oracle query $x' \in G$, the query answer will be equally distributed in both games as the underlying random function $f$ is the same in both games.

In game $\Psi$, an $\mathcal{O}^g$-oracle query, $x''\in G$, will have a uniformly random answer as $g$ is a random function. Likewise, in game $\tilde{\Psi}$, an $\mathcal{O}^g$-oracle query, $x'' \in G$, will have a uniformly random answer as $h$ is a random function.

Consider now the permutation oracle $P= \mathcal{F}_{g,f,f,g}(x\cdot k)\cdot k$. We consider a query-answer pair $\langle x, y \rangle \in T_P$ for $x,y\in G^2$.

In both games, $x^R\cdot k^R$ will be the input to the first round function, which is $g$. In game $\tilde{\Psi}$ the output is always a uniformly random element, newly selected by $g$. In game $\Psi$, if $x^R\cdot k^R$ has already been queried to the $\mathcal{O}^g$-oracle, the output of the round function is the corresponding oracle answer, else it is a uniformly random element, newly selected by $g$.  As the former event in game $\Psi$ never occurs because the event \textbf{BG1} never occurs, the distributions are equal.

As both games have access to the same random function $f$, the second and third round function outputs will have equal distributions.

In both games, $y^L\cdot (k^L)^{-1}$ will be the input to the fourth round function, which is again $g$. In game $\tilde{\Psi}$ the output is always a uniformly random element, newly selected by $g$, unless $y^L\cdot (k^L)^{-1} = x^R\cdot k^R$, in which case the output is equal to the output of the first round function. In game $\Psi$, if $x^R\cdot k^R$ has already been queried to the $\mathcal{O}^g$-oracle, but not as input to the first round function, the output of the round function is the corresponding oracle answer. If $y^L\cdot (k^L)^{-1} = x^R\cdot k^R$, then the output is equal to the output of the first round function, else it is a uniformly random element newly selected by $g$. As the former event in game $\Psi$ never occurs because the event \textbf{BG2} never occurs, the distributions are equal.

As $\mathcal{A}$ does not ask a query if it can determine the answer based on previous queries, we see that the inverse permutation oracle, using $P^{-1}$, yields analogous distributions. Thus, the distributions for the two games must be equal.
\end{proof}

Let us show that the distributions of $T_{\tilde{\Psi}}$ and $T_{\tilde{R}}$ are identical, unless the same value is input to $f$ on two separate occasions. Here we also define when a key is "bad" as we did above, but altered such that it pertains to our current oracles.

\begin{defn}
For every specific key $k=(k^L,k^R)\in_R G^2$ and function $g\in_R \mathfrak{F}_{G\rightarrow G}$, define $\textsf{Bad}(k,g)$ to be the set of all possible and consistent $\mathcal{A}$-transcripts $\sigma$ satisfying at least one of the following events:
\begin{itemize}
\item[\textbf{B1}:] $\exists 1\leq i < j \leq q_c$, such that
\begin{align*}
x_i^L\cdot k^L \cdot g(x_i^R\cdot k^R) = x_j^L \cdot k^L \cdot g(x_j^R \cdot k^R)
\end{align*}
\item[\textbf{B2}:] $\exists 1\leq i < j \leq q_c$, such that
\begin{align*}
y_i^R\cdot (k^R)^{-1} \cdot \left(g(y_i^L\cdot (k^L)^{-1})\right)^{-1} = y_j^R\cdot (k^R)^{-1} \cdot \left(g(y_j^L\cdot (k^L)^{-1})\right)^{-1}
\end{align*}
\item[\textbf{B3}:] $\exists 1\leq i , j \leq q_c$, such that
\begin{align*}
x_i^L\cdot k^L \cdot g(x_i^R\cdot k^R) = y_j^R\cdot (k^R)^{-1} \cdot \left(g(y_j^L\cdot (k^L)^{-1})\right)^{-1}
\end{align*}
\item[\textbf{B4}:] $\exists 1\leq i \leq q_c, 1\leq j \leq q_f$, such that
\begin{align*}
x_i^L\cdot k^L \cdot g(x_i^R\cdot k^R) = x'_j
\end{align*}
\item[\textbf{B5}:] $\exists 1\leq i \leq q_c, 1\leq j \leq q_f$, such that
\begin{align*}
y_i^R\cdot (k^R)^{-1} \cdot \left(g(y_i^L\cdot (k^L)^{-1})\right)^{-1} = x'_j
\end{align*}
\end{itemize}
\end{defn}

\begin{lem}\label{Lem42}
Let $k=(k^L,k^R)\in_R G^2$. For any possible and consistent $\mathcal{A}$-transcript $\sigma$, we have that
\begin{align*}
\underset{k,g}{Pr}\left[ \sigma \in \textsf{Bad}(k,g)\right] \leq \left( q_c^2+2q_fq_c + 2\cdot \begin{pmatrix} q_c \\ 2 \end{pmatrix} \right)\cdot |G|^{-1}.
\end{align*}
\end{lem}

\begin{proof}
We have that $\sigma \in \textsf{Bad}(k,g)$ if it satisfies a $\textit{\textbf{Bi}}$ for some $\textbf{i}=\lbrace1,\ldots,5\rbrace$. Using that $k^L,k^R$ are uniform and independently chosen, and $g\in_R \mathfrak{F}_{G\rightarrow G}$, we may achieve an upper bound on the individual event probabilities, and then use the union bound.

There are $\begin{pmatrix} q_c \\ 2 \end{pmatrix}$ many ways of picking $i,j$ such that $1\leq i< j \leq q_c$, also, $q_fq_c$ many ways of picking $i,j$ such that $1\leq i \leq q_c, 1 \leq j \leq q_f$, and $q_c^2$ many ways of picking $i,j$ such that $1 \leq i,j\leq q_c$. The probability that two elements chosen from $G$ are equal is $|G|^{-1}$, so we may bound each event accordingly and achieve, using the union bound, that
\begin{footnotesize}
\begin{align*}
\underset{k,g}{Pr}\left[ \sigma \in \textsf{Bad}(k,g)\right] &= \underset{k,g}{Pr}\left[ \bigvee_{i=1}^5 \textit{\textbf{Bi}} \text{ occurs } | \sigma \right] \\
&\leq \sum_{i=1}^5 \underset{k,g}{Pr}\left[ \textit{\textbf{Bi}} \text{ occurs } | \sigma \right] \\
&\leq \begin{pmatrix} q_c \\ 2 \end{pmatrix}\cdot |G|^{-1} + \begin{pmatrix} q_c \\ 2 \end{pmatrix}\cdot |G|^{-1} + q_c^2 \cdot |G|^{-1} + q_fq_c \cdot |G|^{-1} + q_fq_c \cdot |G|^{-1} \\
&= \left( q_c^2 + 2q_fq_c + 2\begin{pmatrix} q_c \\ 2 \end{pmatrix} \right) \cdot |G|^{-1}.
\end{align*}
\end{footnotesize}
\end{proof}

\begin{lem}\label{Lem43}
Let $\sigma$ be a possible and consistent $\mathcal{A}$-transcript, then
\begin{align*}
\underset{\tilde{\Psi}}{Pr}\left[ T_{\tilde{\Psi}} = \sigma | \sigma \not\in \textsf{Bad}(k,g)\right] = \underset{\tilde{R}}{Pr}\left[ T_{\tilde{R}} = \sigma \right].
\end{align*}
\end{lem}

The following proof is based on the proof in \citep{GentryRamzan} which refers to \citep{NaorReingold} for the first part of their argument. We need the generalization of this argument and so also include it.

\begin{proof}
Since $\sigma$ is a possible $\mathcal{A}$-transcript, we have for all $1\leq i \leq q_c, 1\leq j \leq q_f, 1\leq k \leq q_g$:
\begin{align*}
C_\mathcal{A}\left[ \lbrace \langle x_1,y_1 \rangle, \ldots, \langle x_{i},y_{i} \rangle \rbrace_P, \lbrace \langle x'_1,y'_1 \rangle, \ldots, \langle x'_{j},y'_{j} \rangle \rbrace_f, \lbrace \langle x''_1,y''_1 \rangle, \ldots, \langle x''_{k},y''_{k} \rangle \rbrace_g \right] \\ \hspace*{70pt}\in \lbrace (+,x_{i+1}), (-,y_{i+1}), (\mathcal{O}^f, x'_{j+1}), (\mathcal{O}^g,x''_{k+1})\rbrace .
\end{align*}
Therefore, $T_{\tilde{R}}=\sigma$ if and only if $\forall 1\leq i \leq q_c, \forall 1\leq j \leq q_f$, and $\forall 1\leq k \leq q_g$, the $i,j,k$'th respective answers $\tilde{R}$ gives are $y_i$ or $x_i$, and $x'_j$ and $x''_k$, respectively. As $\mathcal{A}$ never repeats any part of a query, we have, by the definition of $\tilde{R}$, that the $i$'th cipher-query answer is an independent and uniform element of $G^2$, and as $f$ and $g$ were modelled as random function oracles, so too will their oracle outputs be independent and uniform elements of $G$. Hence,
\begin{align*}
\underset{\tilde{R}}{Pr}\left[ T_{\tilde{R}} = \sigma \right] = |G|^{-(2q_c+q_f+q_g)}.
\end{align*}

For the second part of this proof, we fix $k,g$ such that $\sigma \not\in \textsf{Bad}(k,g)$ and seek to compute $\underset{f,h}{Pr}\left[ T_{\tilde{\Psi}} = \sigma \right]$.
Since $\sigma$ is a possible $\mathcal{A}$-transcript, we have that $T_{\tilde{\Psi}}= \sigma$ if and only if
\begin{itemize}
\item $y_i = \mathcal{F}_{g,f,f,g}(x_i\cdot k)\cdot k$ for all $1\leq i \leq q_c$,
\item $y'_j = f(x'_j)$ for all $1\leq j \leq q_f$, and
\item $y''_k = g(x''_k)$ for all $1\leq k \leq q_g$ (note that $g=h$ here.)
\end{itemize}
If we define
\begin{align*}
X_i &:= x_i^L\cdot k^L \cdot g(x_i^R\cdot k^R) \\
Y_i &:= y_i^R \cdot (k^R)^{-1} \cdot \left(g(y_i^L\cdot (k^L)^{-1})\right)^{-1},
\end{align*}
then $(y_i^L, y_i^R) = \tilde{\Psi}(x_i^L,x_i^R)$  if and only if
\begin{center}
$k^R\cdot f(X_i) = (x_i^R)^{-1} \cdot Y_i$ \hspace*{5pt}  and \hspace*{5pt}  $X_i\cdot f(Y_i) = y_i^L \cdot (k^L)^{-1}$,
\end{center}
where the second equality of the latter is equivalent to $(k^L)^{-1}\cdot (f(Y_i))^{-1} = (y_i^L)^{-1}\cdot X_i$. Observe that, for all $1 \leq i < j \leq q_c$, $X_i \neq X_j$ (by \textit{\textbf{B1}}) and $Y_i \neq Y_j$ (by \textit{\textbf{B2}}.) Similarly, $1 \leq i < j \leq q_c$, $X_i \neq Y_j$ (by \textit{\textbf{B3}}.) Also, for all $1 \leq i \leq q_c$ and for all $1 \leq j \leq q_f$, $x'_j \neq X_i$ (by \textit{\textbf{B4}}) and $x'_j \neq Y_i$ (by \textit{\textbf{B5}}.) Hence, $\sigma \not\in \textsf{Bad}(k,g)$ implies that all inputs to $f$ are distinct. This then implies that $Pr_{f,h}\left[ T_{\tilde{\Psi}} = \sigma \right] = |G|^{-(2q_c+q_f+q_g)}$ as $h$ was also modelled as a random function, independent from $g$.
Thus, as we assumed that $k$ and $g$ were chosen such that $\sigma \not\in \textsf{Bad}(k,g)$,
\begin{align*}
\underset{\tilde{\Psi}}{Pr}\left[ T_{\tilde{\Psi}} = \sigma | \sigma \not\in \textsf{Bad}(k,g)\right] &= |G|^{-(2q_c+q_f+q_g)} =  \underset{\tilde{R}}{Pr}\left[ T_{\tilde{R}} = \sigma \right].
\end{align*}
\end{proof}

\subsection{Proof of Theorem~\ref{GentryRamzanMain}}
To complete the proof of Theorem~\ref{GentryRamzanMain}, we combine the above lemmas into the following probability estimation.

\begin{proof}[Proof of Theorem~\ref{GentryRamzanMain}]
Let $\Gamma$ be the set of all possible and consistent $\mathcal{A}$-transcripts $\sigma$ such that $\mathcal{A}(\sigma)=1$. In the following, we ease notation, for the sake of the reader. We let $\textsf{BadG}(k)$ be denoted by $BadG$ and $\textsf{Bad}(k,g)$ by $Bad$. Furthermore, we abbreviate inconsistency as $incon.$. Let us consider the cases between $\Psi,\tilde{\Psi}$ and $\tilde{R}$.

\begin{footnotesize}
\begin{align*}
&\left| Pr_{\Psi}\left[ C_\mathcal{A}(T_\Psi)=1\right] - Pr_{\tilde{\Psi}}\left[ C_\mathcal{A}(T_{\tilde{\Psi}})=1\right] \right| \\
&\leq \left| \sum_{\sigma\in\Gamma} \left( Pr_{\Psi}\left[ T_\Psi = \sigma \right] - Pr_{\tilde{\Psi}}\left[ T_{\tilde{\Psi}} = \sigma \right]\right) \right| + Pr_{\tilde{\Psi}}\left[ T_{\tilde{\Psi}} \hspace*{4pt} incon.\right] \\
&\leq \sum_{\sigma\in \Gamma} \left| Pr_\Psi \left[ T_{\Psi} = \sigma \mid \sigma \not\in BadG \right] - Pr_{\tilde{\Psi}}\left[ T_{\tilde{\Psi}} = \sigma \right] \right| \cdot Pr_k \left[ \sigma \not\in BadG \right] \\
	&\hspace*{15pt}+ \left| \sum_{\sigma \in \Gamma} \left( Pr_{\Psi} \left[ T_\Psi = \sigma \mid \sigma \in BadG \right] - Pr_{\tilde{\Psi}} \left[ T_{\tilde{\Psi}} = \sigma \right]\right) \cdot Pr_k \left[ \sigma \in BadG \right]\right| \\
		&\hspace*{30pt}+ Pr_{\tilde{\Psi}}\left[ T_{\tilde{\Psi}} \hspace*{4pt} incon. \right] \\
&\leq \left| \sum_{\sigma \in \Gamma} \left( Pr_{\Psi} \left[ T_\Psi = \sigma \mid \sigma \in BadG \right] - Pr_{\tilde{\Psi}} \left[ T_{\tilde{\Psi}} = \sigma \right]\right) \cdot Pr_k \left[ \sigma \in BadG \right]\right| + q_c(q_c-1)|G|^{-1},
\end{align*}
\end{footnotesize}
where we in the last estimate used Lemma~\ref{NotBadPsitoBarPsi} and a consideration of the maximal amount of possible inconsistent pairs.

At the same time,
\begin{footnotesize}
\begin{align*}
&\left| Pr_{\tilde{\Psi}}\left[ C_\mathcal{A}(T_{\tilde{\Psi}})=1\right] - Pr_{\tilde{R}}\left[ C_\mathcal{A}(T_{\tilde{R}})=1\right] \right| \\
&\leq \left| \sum_{\sigma\in\Gamma} \left( Pr_{\tilde{\Psi}}\left[ T_{\tilde{\Psi}} = \sigma \right] - Pr_{\tilde{R}}\left[ T_{\tilde{R}} = \sigma \right]\right) \right| + Pr_{\tilde{R}}\left[ T_{\tilde{R}} \hspace*{4pt} incon.\right] + Pr_{\tilde{\Psi}}\left[ T_{\tilde{\Psi}} incon. \right] \\
&\leq \sum_{\sigma\in \Gamma} \left| Pr_{\tilde{\Psi}} \left[ T_{\tilde{\Psi}} = \sigma \mid \sigma \not\in Bad \right] - Pr_{\tilde{R}}\left[ T_{\tilde{R}} = \sigma \right] \right| \cdot Pr_k \left[ \sigma \not\in Bad \right] \\
	&\hspace*{15pt} + \left| \sum_{\sigma \in \Gamma} \left( Pr_{\tilde{\Psi}} \left[ T_{\tilde{\Psi}} = \sigma \mid \sigma \in Bad \right] - Pr_{\tilde{R}} \left[ T_{\tilde{R}} = \sigma \right]\right) \cdot Pr_k \left[ \sigma \in Bad \right]\right| \\
		&\hspace*{30pt}+ Pr_{\tilde{R}}\left[ T_{\tilde{R}} \hspace*{4pt} incon. \right] + Pr_{\tilde{\Psi}}\left[ T_{\tilde{\Psi}} incon. \right] \\
&\leq \left| \sum_{\sigma \in \Gamma} \left( Pr_{\tilde{\Psi}} \left[ T_{\tilde{\Psi}} = \sigma \mid \sigma \in Bad \right] - Pr_{\tilde{R}} \left[ T_{\tilde{R}} = \sigma \right]\right) \cdot Pr_k \left[ \sigma \in Bad \right]\right| + \begin{pmatrix} q_c \\ 2 \end{pmatrix}|G|^{-2} + 2\begin{pmatrix} q_c \\ 2 \end{pmatrix}|G|^{-1},
\end{align*}
\end{footnotesize}
where we in the last estimate used Lemma~\ref{Lem43} and the proof of Lemma~\ref{Lem37}.

Let us use the above in a temporary estimate,
\begin{footnotesize}
\begin{align}
&\left| Pr_\Psi\left[ C_\mathcal{A}(T_\Psi)=1\right] - Pr_R\left[C_\mathcal{A}(T_R)=1\right] \right| \nonumber \\
&= \left| Pr_\Psi\left[ C_\mathcal{A}(T_\Psi)=1\right] - Pr_{\tilde{\Psi}}\left[C_\mathcal{A}(T_{\tilde{\Psi}})=1\right] \right| \nonumber \\
	&\hspace*{40pt}+ \left| Pr_{\tilde{\Psi}}\left[ C_\mathcal{A}(T_{\tilde{\Psi}})=1\right] - Pr_{\tilde{R}}\left[C_\mathcal{A}(T_{\tilde{R}})=1\right] \right| \nonumber \\
		&\hspace*{80pt}+ \left| Pr_{\tilde{R}}\left[ C_\mathcal{A}(T_{\tilde{R}})=1\right] - Pr_R\left[C_\mathcal{A}(T_R)=1\right] \right| \nonumber \\
&\leq \left| \sum_{\sigma \in \Gamma} \left( Pr_{\Psi} \left[ T_\Psi = \sigma \mid \sigma \in BadG \right] - Pr_{\tilde{\Psi}} \left[ T_{\tilde{\Psi}} = \sigma \right]\right) \cdot Pr_k \left[ \sigma \in BadG \right]\right| + q_c(q_c-1)|G|^{-1} \nonumber\\
	&\hspace*{30pt}+ \left| \sum_{\sigma \in \Gamma} \left( Pr_{\tilde{\Psi}} \left[ T_{\tilde{\Psi}} = \sigma \mid \sigma \in Bad \right] - Pr_{\tilde{R}} \left[ T_{\tilde{R}} = \sigma \right]\right) \cdot Pr_k \left[ \sigma \in Bad \right]\right| + \begin{pmatrix} q_c \\ 2 \end{pmatrix}|G|^{-2} + 2\begin{pmatrix} q_c \\ 2 \end{pmatrix}|G|^{-1} \nonumber\\
		&\hspace*{60pt}+ \begin{pmatrix} q_c \\ 2 \end{pmatrix}|G|^{-2} \label{GenRamMainEstimate},
\end{align}
\end{footnotesize}
where we in the last estimate also used Lemma~\ref{Lem37}.

We may assume WLOG that
\begin{footnotesize}
\begin{align*}
\sum_{\sigma \in \Gamma} Pr_\Psi\left[ T_\Psi = \sigma \mid \sigma \in BadG \right] \cdot Pr_k\left[ \sigma \in BadG \right] \leq \sum_{\sigma \in \Gamma} Pr_{\tilde{\Psi}}\left[ T_{\tilde{\Psi}} = \sigma\right] \cdot Pr_k\left[ \sigma \in BadG \right]
\end{align*}
\end{footnotesize}
and likewise,
\begin{footnotesize}
\begin{align*}
\sum_{\sigma \in \Gamma} Pr_{\tilde{\Psi}}\left[ T_{\tilde{\Psi}} = \sigma \mid \sigma \in BadG \right] \cdot Pr_k\left[ \sigma \in BadG \right] \leq \sum_{\sigma \in \Gamma} Pr_{\tilde{R}}\left[ T_{\tilde{R}} = \sigma\right] \cdot Pr_k\left[ \sigma \in BadG \right],
\end{align*}
\end{footnotesize}
such that by Lemma~\ref{Lem39}, respectively Lemma~\ref{Lem42}, we get the following continued estimate from (\ref{GenRamMainEstimate}), using the triangle inequality and that $|\Gamma|\leq |G|^{2q_c+q_f+q_g}$ (every combination of query elements).
\begin{footnotesize}
\begin{align*}
&\left| Pr_\Psi\left[ C_\mathcal{A}(T_\Psi)=1\right] - Pr_R\left[C_\mathcal{A}(T_R)=1\right] \right| \\
&\leq 2 \sum_{\sigma \in \Gamma} Pr_{\tilde{\Psi}} \left[T_{\tilde{\Psi}} = \sigma \right]\cdot Pr_k \left[ \sigma \in BadG \right] + 2q_c(q_c-1)|G|^{-1} \\
	&\hspace*{30pt}+ 2\sum_{\sigma \in \Gamma} Pr_{\tilde{R}}\left[ T_{\tilde{R}}= \sigma\right] \cdot Pr_k \left[\sigma \in Bad \right] \\
		&\hspace*{60pt}+ 2\begin{pmatrix} q_c \\ 2 \end{pmatrix}|G|^{-2} \\
&\leq 2|\Gamma|\cdot |G|^{-(2q_c+q_f+q_g)}\cdot \max_{\sigma \in \Gamma} Pr_k \left[ \sigma \in BadG \right] + 2q_c(q_c-1)|G|^{-1} \\
	&\hspace*{30pt}+ 2|\Gamma|\cdot |G|^{-(2q_c+q_f+q_g)}\cdot\max_{\sigma\in\Gamma} Pr_k \left[ \sigma \in Bad \right] \\
		&\hspace*{60pt}+ 2\begin{pmatrix} q_c \\ 2 \end{pmatrix}|G|^{-2} \\
&\leq 4q_gq_c\cdot |G|^{-1} + 2q_c(q_c-1)|G|^{-1} + 2\left(q_c^2 + 2q_fq_c + 2\begin{pmatrix} q_c \\ 2 \end{pmatrix}\right)|G|^{-1} + 2\begin{pmatrix} q_c \\ 2 \end{pmatrix}|G|^{-2} \\
&= (2q_c^2+4q_gq_c + 4q_fq_c + 2q_c^2-2q_c)|G|^{-1} + 2\begin{pmatrix} q_c \\ 2 \end{pmatrix}\left(2|G|^{-1} + |G|^{-2}\right).
\end{align*}
\end{footnotesize}
\end{proof}

If we denote the total amount of queries as $q=q_c+q_f+q_g$, then we may quickly estimate and reword the main theorem as:

\begin{thm}
Let $f,g$ be modelled as random oracles, let $k=(k^L,k^R)\in_R G^2$, let $\Psi_k^{f,g}(x)=\mathcal{F}_{g,f,f,g}(x\cdot k)\cdot k$, and let $R\in_R \mathfrak{P}_{G^2 \rightarrow G^2}$. Then, for any $4$-oracle adversary $\mathcal{A}$, with at most $q$ total queries, we have
\begin{align*}
\left| Pr\left[ \mathcal{A}^{\Psi,\Psi^{-1},f,g} = 1\right] - Pr\left[ \mathcal{A}^{R,R^{-1},f,g} = 1\right] \right| \leq 2(3q^2-2q)|G|^{-1} + (q^2-q)|G|^{-2}.
\end{align*}
\end{thm}

\begin{proof}
Given Theorem~\ref{GentryRamzanMain}, we get, by using that $q_f,q_g\geq 0$,
\begin{align*}
&q_c^2 +2q_fq_c + 2q_gq_c + q_c^2 - q_c \\ 
	&= 2(q_c^2 + q_fq_c + q_gq_c) - q_c \\
	&\leq 2(q_c^2 + q_fq_c + q_gq_c) + (2(q_f+q_g)^2 + 2(q_fq_c+q_gq_c) -q_f-q_g)-q_c \\
	&= 2(q_c^2 + 2q_fq_c  + q_f^2 + 2q_fq_g + 2q_gq_c+ q_g^2) - (q_c+q_f+q_g) \\
	&= 2(q_c+q_f+q_g)^2 - (q_c+q_f+q_g) \\
	&= 2q^2-q.
\end{align*}
As $2\cdot\begin{pmatrix} q_c \\ 2 \end{pmatrix} =q_c^2-q_c \leq q^2-q$, we get the final estimate by some reordering.
\end{proof}

\section{Conclusion}\label{ConclusionSection}
We generalized the Even and Mansour scheme as well as the Feistel cipher to work over arbitrary groups and proved that classical results pertain to the group versions. Based on the work in \citep{Gorjan}, we hope that this opens avenues to proving that classical schemes may be made quantum secure by generalizing them to certain groups. For further work, we suggest generalizing other classical schemes and using the underlying group structures to do Hidden Shift reductions.

The author would like to thank his thesis advisor Gorjan Alagic for the topic, enlightening questions and answers, as well as the encouragements along the way.
The author would also like to thank the Department of Mathematical Sciences, at the University of Copenhagen, for lending their facilities during the writing process.

\pagebreak
\bibliographystyle{alpha}
\bibliography{Articlebib}

\clearpage
\appendix
\newpage
\section{Super Pseudorandomness of the Group EM Scheme}
In the following, we assume that the adversary $\mathcal{A}$ is unbounded computationally, but may only make polynomially many queries to the $E/D$- and $P/P^{-1}$-oracles, where all oracles act as black boxes and $P$ is a truly random permutation. We intend to play the "pseudorandom or random permutation game": $\mathcal{A}$ is given an encryption oracle $E$ (with related decryption oracle $D$) which is randomly chosen with equal probability from the following two options:
\begin{enumerate}
\item A random key $k\in_R G$ is chosen uniformly and used to encrypt as $E(m)=E_k(m)=P(m\cdot k )\cdot k$, or
\item A random permutation $\pi\in_R \mathfrak{P}_{G \rightarrow G}$ is chosen and used to encrypt as $E(m)=\pi(m)$.
\end{enumerate}
The adversary wins the game if it can distinguish how $E$ was chosen, with probability significantly better than $1/2$. More explicitly, we wish to prove the following for the group Even-Mansour scheme.

\begin{thm}\label{PseudoEMbounded}
Assume $P\in_R\mathfrak{P}_{G\rightarrow G}$ and let the key $k\in_R G$. For any probabilistic adversary $\mathcal{A}$, limited to polynomially many $E/D$- and $P/P^{-1}$-oracle queries, the adversarial advantage of $\mathcal{A}$ is bounded by
\begin{align}\label{AdvA}
\text{Adv}(\mathcal{A}) \defeq \left| Pr\left[ \mathcal{A}_{E_k,D_k}^{P,P^{-1}} = 1 \right] - Pr\left[\mathcal{A}_{\pi,\pi^{-1}}^{P,P^{-1}} = 1 \right]\right| = \mathcal{O}\left(\frac{st}{|G|}\right).
\end{align}
where $s$ is the total number of $E/D$-queries and $t$ is the total number of $P/P^{-1}$-queries, i.e. the success probability is negligible.
\end{thm}

\begin{proof}
We may assume that $\mathcal{A}$ is deterministic (in essence, being unbounded computationally affords $\mathcal{A}$ the possibility of derandomizing its strategy by searching all its possible random choices and picking the most effective choices after having computed the effectiveness of each choice. For an example, see \citep{DingDong}.) We may also assume that $\mathcal{A}$ never queries a pair in $S_s$ or $T_t$ more than once, where $S_i$ and $T_i$ are the sets of $i$ $E/D$- and $P/P^{-1}$-queries, respectively. Let us define two main games, that $\mathcal{A}$ could play, through oracle interactions (see next page for the explicit game descriptions.)

Note that the steps in italics have no impact on the response to $\mathcal{A}$'s queries, we simply continue to answer the queries and only note if the key turns bad, i.e. we say that a key $k$ is \textbf{bad w.r.t. the sets $S_s$ and $T_t$} if there exist $i,j$ such that either $m_i \cdot k = x_j$ or $c_i \cdot k^{-1} = y_j$, and $k$ is \textbf{good} otherwise. There are at most $\frac{2st}{|G|}$ bad keys.

\textbf{Game R}: We consider the random game which corresponds to the latter probability in (\ref{AdvA}), i.e.
\begin{align*}
P_R := Pr\left[ \mathcal{A}_{\pi,\pi^{-1}}^{P,P^{-1}}=1\right].
\end{align*}

 From the definition of \textbf{Game R}, we see that, letting $Pr_R$ denote the probability when playing \textbf{Game R},
\begin{align}\label{PR}
Pr_R \left[ \mathcal{A}_{E,D}^{P,P^{-1}}=1\right] = P_R,
\end{align}
as we are simply giving uniformly random answers to each of $\mathcal{A}$'s queries.

\newpage
\begin{footnotesize}\label{GamesXandR}
\noindent\textbf{Notation:} We let $S^1_i = \lbrace m | (m,c)\in S_{i} \rbrace, \hspace*{5pt} S^2_i = \lbrace c | (m,c)\in S_{i} \rbrace, T^1_i = \lbrace x | (x,y)\in T_{i} \rbrace,$ and $\hspace*{5pt} T^2_i = \lbrace y | (x,y)\in T_{i} \rbrace.$
\end{footnotesize}
\begin{small}
\hrule\noindent
\begin{minipage}[t]{0.48\textwidth}
\vspace*{0.01\textheight}\textbf{GAME R:} Initially, let $S_0$ and $T_0$ be empty and flag unset. Choose $k\in_R G$, then answer the $i+1$'st query as follows: \linebreak
\vspace*{0.0005\textheight}

\textbf{$E$-oracle query with $m_{i+1}$:} \\
 \textbf{1.} Choose $c_{i+1}\in_R G\setminus S^2_i$. \\
 \textbf{2.} \textit{If $P(m_{i+1}\cdot k)\in T^2_i$, or $P^{-1}(c_{i+1}\cdot k^{-1})\in T^1_i$, then set flag to \textbf{bad}.} \\
 \textbf{3.} Define $E(m_{i+1})=c_{i+1}$ (and thereby also $D(c_{i+1})=m_{i+1}$) and return $c_{i+1}$. \\

\vspace*{0.023\textheight}
\textbf{$D$-oracle query with $c_{i+1}$:} \\
 \textbf{1.} Choose $m_{i+1}\in_R G\setminus S^1_i$. \\
 \textbf{2.} \textit{If $P^{-1}(c_{i+1}\cdot k^{-1})\in T^1_i$, or $P(m_{i+1}\cdot k)\in T^2_i$, then set flag to \textbf{bad}.} \\
 \textbf{3.} Define $D(c_{i+1}) = m_{i+1}$ (and thereby also $E(m_{i+1})=c_{i+1}$) and return $m_{i+1}$. \\

\vspace*{0.045\textheight}
\textbf{$P$-oracle query with $x_{i+1}$:} \\
 \textbf{1.} Choose $y_{i+1}\in_R G\setminus T^2_i$. \\
 \textbf{2.} \textit{If $E(x_{i+1}\cdot k^{-1})\in S^2_i$, or $D(y_{i+1}\cdot k)\in S^1_i$, then set flag to \textbf{bad}.} \\
 \textbf{3.} Define $P(x_{i+1}) = y_{i+1}$ (and thereby also $P^{-1}(y_{i+1})=x_{i+1}$) and return $y_{i+1}$. \\

\vspace*{0.023\textheight}
\textbf{$P^{-1}$-oracle query with $y_{i+1}$:} \\
 \textbf{1.} Choose $x_{i+1}\in_R G\setminus T^1_i$. \\
 \textbf{2.} \textit{If $D(y_{i+1}\cdot k)\in S^1_i$, or $E(x_{i+1}\cdot k^{-1})\in S^2_i$, then set flag to \textbf{bad}.} \\
 \textbf{3.} Define $P^{-1}(y_{i+1}) = x_{i+1}$ (and thereby also $P(x_{i+1})=y_{i+1}$) and return $x_{i+1}$.
\end{minipage}\hspace*{0.01\textwidth} \vrule \hspace*{0.01\textwidth}
\begin{minipage}[t]{0.48\textwidth}
\vspace*{0.01\textheight}\textbf{GAME X:} Initially, let $S_0$ and $T_0$ be empty and flag unset. Choose $k\in_R G$, then answer the $i+1$'st query as follows: \linebreak
\vspace*{0.0005\textheight}

\textbf{$E$-oracle query with $m_{i+1}$:} \\
 \textbf{1.} Choose $c_{i+1}\in_R G\setminus S^2_i$. \\
 \textbf{2.} If $P(m_{i+1}\cdot k)\in T^2_i$ then redefine $c_{i+1} := P(m_{i+1}\cdot k)\cdot k$ \textit{and set flag to \textbf{bad}}. Else if $P^{-1}(c_{i+1}\cdot k^{-1})\in T^1_i$, \textit{then set flag to \textbf{bad} and} goto Step 1. \\
 \textbf{3.} Define $E(m_{i+1})=c_{i+1}$ (and thereby also $D(c_{i+1})=m_{i+1}$) and return $c_{i+1}$. \\

\textbf{$D$-oracle query with $c_{i+1}$:} \\
 \textbf{1.} Choose $m_{i+1}\in_R G\setminus S^1_i$. \\
 \textbf{2.} If $P^{-1}(c_{i+1}\cdot k^{-1})\in T^1_i$ then redefine $m_{i+1} := P^{-1}(c_{i+1}\cdot k^{-1})\cdot k^{-1}$ \textit{and set flag to \textbf{bad}}. Else if $P(m_{i+1}\cdot k)\in T^2_i$, \textit{then set flag to \textbf{bad} and} goto Step 1. \\
 \textbf{3.} Define $D(c_{i+1}) = m_{i+1}$ (and thereby also $E(m_{i+1})=c_{i+1}$) and return $m_{i+1}$. \\

\textbf{$P$-oracle query with $x_{i+1}$:} \\
 \textbf{1.} Choose $y_{i+1}\in_R G\setminus T^2_i$. \\
 \textbf{2.} If $E(x_{i+1}\cdot k^{-1})\in S^2_i$ then redefine $y_{i+1} := E(x_{i+1}\cdot k^{-1})\cdot k^{-1}$ \textit{and set flag to \textbf{bad}}. Else if $D(y_{i+1}\cdot k)\in S^1_i$, \textit{then set flag to \textbf{bad} and} goto Step 1. \\
 \textbf{3.} Define $P(x_{i+1}) = y_{i+1}$ (and thereby also $P^{-1}(y_{i+1})=x_{i+1}$) and return $y_{i+1}$. \\

\textbf{$P^{-1}$-oracle query with $y_{i+1}$:} \\
 \textbf{1.} Choose $x_{i+1}\in_R G\setminus T^1_i$. \\
 \textbf{2.} If $D(y_{i+1}\cdot k)\in S^1_i$ then redefine $x_{i+1} := D(y_{i+1}\cdot k)\cdot k$ \textit{and set flag to \textbf{bad}}. Else if $E(x_{i+1}\cdot k^{-1})\in S^2_i$, \textit{then set flag to \textbf{bad} and} goto Step 1. \\
 \textbf{3.} Define $P^{-1}(y_{i+1}) = x_{i+1}$ (and thereby also $P(x_{i+1})=y_{i+1}$) and return $x_{i+1}$.
\vspace*{0.02\textheight}
\end{minipage}
\end{small}

\newpage

\textbf{Game X}: Consider the experiment which corresponds to the game played in the prior probability in (\ref{AdvA}) and define this probability as
\begin{align*}
P_X := Pr\left[ \mathcal{A}_{E_k,D_k}^{P,P^{-1}}=1\right].
\end{align*}
We define \textbf{Game X}, as outlined above. Note that again the parts in italics have no impact on the response to $\mathcal{A}$'s queries, however, this time, when a key becomes \emph{bad}, we choose a new random value repeatedly for the response until the key is no longer \emph{bad}, and then reply with this value. Intuitively, \textbf{Game X} behaves like \textbf{Game R} except that \textbf{Game X} checks for consistency as it does not want $\mathcal{A}$ to win on some collision. It is non-trivial to see that, letting $Pr_X$ denote the probability when playing \textbf{Game X},
\begin{align}\label{PX}
Pr_X \left[ \mathcal{A}_{E,D}^{P,P^{-1}}=1\right] =  P_X.
\end{align}
The proof is given in Appendix~\ref{ExplainX}.

We have defined both games in such a way that their outcomes differ only in the event that a key turns \textit{bad}. Thus, any circumstance which causes a difference in the instructions carried out by the games, will also cause both games to set the flag to \emph{bad}. Let $BAD$ denote the event that the flag gets set to \emph{bad} and the case that the flag is not set to \emph{bad} by $\neg BAD$, then the two following lemmas follow from the previous statement.

\begin{lem}\label{BadisBad}
$Pr_R\left[ BAD \right] = Pr_X\left[ BAD \right]$ and $Pr_R\left[ \neg BAD \right] = Pr_X\left[ \neg BAD \right]$.
\end{lem}

\begin{lem}\label{Notbadisnotbad}
$Pr_R\left[ \mathcal{A}_{E,D}^{P,P^{-1}}=1| \neg BAD \right] = Pr_X\left[ \mathcal{A}_{E,D}^{P,P^{-1}}=1| \neg BAD \right]$.
\end{lem}

Using these two lemmas we are able to prove the lemma:

\begin{lem}
$\text{Adv}(\mathcal{A}) \leq Pr_R\left[ BAD \right]$.
\end{lem}

This is because, using (\ref{PR}), (\ref{PX}), and lemmas \ref{BadisBad} and \ref{Notbadisnotbad}, 
\begin{align*}
\text{Adv}(\mathcal{A}) &= |P_X - P_R| \\
						&= \left| Pr_X\left[ \mathcal{A}_{E,D}^{P,P^{-1}}=1 \right] - Pr_R\left[ \mathcal{A}_{E,D}^{P,P^{-1}}=1 \right] \right| \\
						&= | Pr_X\left[ \mathcal{A}_{E,D}^{P,P^{-1}}=1 | \neg BAD \right]\cdot Pr_X\left[\neg BAD \right] \\
							&\hspace*{25pt}+ Pr_X\left[ \mathcal{A}_{E,D}^{P,P^{-1}}=1 | BAD \right]\cdot Pr_X\left[ BAD \right] \\
							&\hspace*{45pt} - Pr_R\left[ \mathcal{A}_{E,D}^{P,P^{-1}}=1 | \neg BAD \right]\cdot Pr_R\left[\neg BAD \right] \\
							&\hspace*{65pt}- Pr_R\left[ \mathcal{A}_{E,D}^{P,P^{-1}}=1 | BAD \right]\cdot Pr_R\left[ BAD \right] | \\
						&= \left| Pr_R\left[ BAD \right]\cdot \left( Pr_X\left[ \mathcal{A}_{E,D}^{P,P^{-1}}=1 | BAD \right] - Pr_R\left[ \mathcal{A}_{E,D}^{P,P^{-1}}=1 | BAD \right]\right) \right| \\
						&\leq Pr_R\left[ BAD \right].
\end{align*}

Let us now define yet another game, \textbf{Game R'}.
\vspace*{0.015\textheight}
\hrule
\begin{footnotesize}
\begin{minipage}[t]{0.9\textwidth}
\vspace*{0.005\textheight} \textbf{GAME R':} Initially, let $S_0$ and $T_0$ be empty and flag unset. Answer the $i+1$'st query as follows: \linebreak
\textbf{$E$-oracle query with $m_{i+1}$:} \\
 \textbf{1.} Choose $c_{i+1}\in_R G\setminus S^2_i$. \\
 \textbf{2.} Define $E(m_{i+1}):=c_{i+1}$ (and thereby also $D(c_{i+1}):=m_{i+1}$) and return $c_{i+1}$. \\

\textbf{$D$-oracle query with $c_{i+1}$:} \\
 \textbf{1.} Choose $m_{i+1}\in_R G\setminus S^1_i$. \\
 \textbf{2.} Define $D(c_{i+1}) := m_{i+1}$ (and thereby also $E(m_{i+1}):=c_{i+1}$) and return $m_{i+1}$. \\

\textbf{$P$-oracle query with $x_{i+1}$:} \\
 \textbf{1.} Choose $y_{i+1}\in_R G\setminus T^2_i$. \\
 \textbf{2.} Define $P(x_{i+1}) := y_{i+1}$ (and thereby also $P^{-1}(y_{i+1}):=x_{i+1}$) and return $y_{i+1}$. \\

\textbf{$P^{-1}$-oracle query with $y_{i+1}$:} \\
 \textbf{1.} Choose $x_{i+1}\in_R G\setminus T^1_i$. \\
 \textbf{2.} Define $P^{-1}(y_{i+1}) := x_{i+1}$ (and thereby also $P(x_{i+1}):=y_{i+1}$) and return $x_{i+1}$.\\
 
\textit{After all queries have been answered, choose $k\in_R G$. If there exists $(m,c)\in S_{s}$ and $(x,y)\in T_{t}$ such that $k$ becomes bad then set flag to \textbf{bad}.}
\vspace*{0.005\textheight}
\end{minipage}
\end{footnotesize}
\hrule
\vspace*{0.015\textheight}
This game runs as \textbf{Game R} except that it does not choose a key until all of the queries have been answered and then checks for badness of the flag (by checking whether or not the key has become bad). It can be shown that the flag is set to \textbf{\textit{bad}} in \textbf{Game R} if and only if the flag is set to \textbf{\textit{bad}} in \textbf{Game R'} (by a consideration of cases (see Appendix~\ref{App:ReqRR}.)) Hence, we get the following lemma.

\begin{lem}
$Pr_R\left[ BAD \right] = Pr_{R'} \left[ BAD \right]$.
\end{lem}

Using the above lemma, we now only have to bound $Pr_{R'} \left[ BAD \right]$ in order to bound $\text{Adv}(\mathcal{A})$, but as the adversary queries at most $s$ elements to the $E/D$-oracles and at most $t$ elements to the $P/P^{-1}$-oracles, and the key $k$ is chosen uniformly at random from $G$, we have that the probability of choosing a bad key is at most $2st/|G|$, i.e.
\begin{align*}
\text{Adv}(\mathcal{A}) \leq Pr_{R'} \left[ BAD \right] = \mathcal{O}\left( \frac{st}{|G|} \right).
\end{align*}

\end{proof}

Restating the theorem, we get:

\begin{thm}
For any probabilistic adversary $\mathcal{A}$, limited to polynomially many $E/D$- and $P/P^{-1}$-oracle queries, the generalized EM scheme over a group $G$ is a super pseudorandom permutation.
\end{thm}

\newpage
\section{Proof of probability of Game X}\label{ExplainX}
Recall the definition of $S^1_i, S^2_i, T^1_i$ and $T^2_i$ (see p.~\pageref{GamesXandR}.) We write $S_s$ and $T_t$ to denote the final transcripts. We drop the index $i$ if it is understood. We begin by defining \textbf{Game X'}.

\vspace*{0.015\textheight}
\hrule
\begin{footnotesize}
\begin{minipage}[t]{0.9\textwidth}
\vspace*{0.005\textheight}\textbf{GAME X':} Initially, let $S_0$ and $T_0$ be empty. Choose $k\in_R G$, then answer the $i+1$'st query as follows:

\textbf{$E$-oracle query with $m_{i+1}$:} \\
 \textbf{1.} If $P(m_{i+1}\cdot k)\in T^2_i$ return $P(m_{i+1}\cdot k)\cdot k$ \\
 \textbf{2.} Else choose $y_{i+1}\in_R G\setminus T^2_i$, define $P(m_{i+1}\cdot k) = y_{i+1}$, and return $y_{i+1}\cdot k$. \\

\textbf{$D$-oracle query with $c_{i+1}$:} \\
 \textbf{1.} If $P^{-1}(c_{i+1}\cdot k^{-1}) \in T^1_i$, return $P^{-1}(c_{i+1}\cdot k^{-1})\cdot k^{-1}$. \\
 \textbf{2.} Else choose $x_{i+1}\in_R G\setminus T^1_i$, define $P^{-1}(c_{i+1}\cdot k^{-1})=x_{i+1}$, and return $x_{i+1}\cdot k^{-1}$. \\

\textbf{$P$-oracle query with $x_{i+1}$:} \\
 \textbf{1.} If $P(x_{i+1})\in T^2_i$, return $P(x_{i+1})$. \\
 \textbf{2.} Else choose $y_{i+1}\in_R G\setminus T^2_i$, define $P(x_{i+1})=y_{i+1}$, and return $y_{i+1}$. \\

\textbf{$P^{-1}$-oracle query with $y_{i+1}$:} \\
 \textbf{1.} If $P^{-1}(y_{i+1})\in T^1_i$, return $P^{-1}(y_{i+1})$. \\
 \textbf{2.} Else choose $x_{i+1}\in_R G\setminus T^1_i$, define $P^{-1}(y_{i+1})=x_{i+1}$, and return $x_{i+1}$.
\vspace*{0.005\textheight}
\end{minipage}
\end{footnotesize}
\hrule
\vspace*{0.015\textheight}

Notice that the only difference between \textbf{Game X'} and the game defining $P_X$ is that the latter has defined all values for the oracles beforehand while the former "defines as it goes." Still, an adversary cannot tell the difference between playing the \textbf{Game X'} or the game defining $P_X$. Thus, $Pr_{X'}\left[ \mathcal{A}_{E,D}^{P,P^{-1}}=1 \right] = P_X$.

What we wish to show is that
\begin{align*}
Pr_X \left[ \mathcal{A}_{E,D}^{P,P^{-1}}=1\right] = Pr_{X'} \left[ \mathcal{A}_{E,D}^{P,P^{-1}}=1\right],
\end{align*}
i.e. that no adversary $\mathcal{A}$ may distinguish between playing \textbf{Game X} and playing \textbf{Game X'}, even negligibly. We will do this by showing that no adversary $\mathcal{A}$ may distinguish between the outputs given by the two games. As both games begin by choosing a uniformly random key $k$ and as we show that for this value the games are identical, we hereby assume such a key $k$ to be a fixed, but arbitrary, value for the remainder of this proof.

Considering the definitions of \textbf{Game X} and \textbf{Game X'}, we see that the two games define their $E/D$- and $P/P^{-1}$-oracles differently: the former defining both, while the latter defines only the $P/P^{-1}$-oracle and computes the $E/D$-oracle. We show that \textbf{Game X} also answers its $E/D$-oracle queries by referring to $P/P^{-1}$, although not directly.

Given the partial functions $E$ and $P$ in \textbf{Game X}, i.e. functions having been defined for all values up to and including the $i$'th query, define the partial function $\widehat{P}$ as the following.
\begin{align*}
\widehat{P}(x) \defeq 	\begin{cases}
							P(x) & \text{if } P(x) \text{ is defined,} \\
							E(x\cdot k^{-1}) \cdot k^{-1} & \text{if } E(x\cdot k^{-1}) \text{ is defined, and} \\
							\text{undefined} & \text{otherwise.}
						\end{cases}
\end{align*}
Using the above definition, defining a value for $E$ or $P$ implicitly defines a value for $\widehat{P}$. The first question is, whether or not $\widehat{P}$ is well-defined, i.e. whether there are clashes of values (that is, differences between values differing by other than $\cdot k$ (or $\cdot k^{-1}$)) for some $x$ for which both $P(x)$ and $E(x\cdot k^{-1})$ are defined.

\begin{lem}
Let $E$ and $P$ be partial functions arising in \textbf{Game X}, then the partial function $\widehat{P}$ is well-defined.
\end{lem}

\begin{proof}
Proof by induction on the number of "Define" steps in \textbf{Game X} (i.e. steps $E-3, D-3, P-3,$ and $P^{-1}-3$) as these are the steps where $\widehat{P}$ becomes defined. The initial case of the induction proof is trivial as $S_0$ and $T_0$ are empty such that no values may clash. Suppose now that in step $E-3$ we define $E(m)=c$. The only possibility that $\widehat{P}$ becomes ill-defined will occur if the new $E(m)$ value clashes with a prior defined $P(m\cdot k)$ value: If $P(m \cdot k)$ was not defined, then no clashes can arise. If $P(m\cdot k)$ was defined, then by step $E-2$, the value is $E(m)\cdot k^{-1}$, such that there is no clash.

For $D-3$, the argument is similar as $E(m)$ will become defined as well. Although, for the case where $P(m\cdot k)$ is defined, step $D-2$ forces a new uniformly random value of $m$ to be chosen until no clash occurs.

Analogously, for $P$ and $P^{-1}$, no clashes will arise, hence, $\widehat{P}$ must be well-defined.
\end{proof}

We may also consider $\widehat{P}$ in \textbf{Game X'}, in the sense that when we define a value for $P$ in the game, we implicitly define a value for $\widehat{P}$ where $\widehat{P}(x)=P(x)$ as $E(x\cdot k^{-1})=P(x)$ in \textbf{Game X'}.

We wish now to show that the oracle query-answers of $E, D, P,$ and $P^{-1}$ in \textbf{Game X}, expressed in terms of $\widehat{P}$, correspond exactly to those in \textbf{Game X'}.

\textbf{Case 1: $E$-oracle query.} Beginning with \textbf{Game X}, we first note that \textbf{Game X} never defines $E(m)$ unless $m$ has been queried to the $E$-oracle, or alternately, the $D$-oracle has been queried with a $c$ such that $E(m)=c$. However, as $\mathcal{A}$ never repeats a query if it can guess the answer, i.e. never queries any part of an already defined $E/D$-oracle pair, we may assume that $E(m)$ is undefined when $m$ is queried. Therefore, we see that concurrently with $m$ being queried, we have that $\widehat{P}(m\cdot k)$ will be defined if and only if $P(m\cdot k)$ is defined, and $\widehat{P}(m\cdot k) = P(m\cdot k)$. Let us consider the two cases: when $\widehat{P}(m\cdot k)$ is defined and when it is undefined.
\begin{itemize}
\item[\textit{Case 1a}:] When $\widehat{P}(m\cdot k)$ is defined, then \textbf{Game X} returns $c = \widehat{P}(m\cdot k)\cdot k$. Setting $E(m)=c$ leaves $\widehat{P}$ unchanged, i.e. the value $\widehat{P}(m\cdot k)$ remains the same, unlike the next case.
\item[\textit{Case 1b}:] \begin{sloppypar} When $\widehat{P}(m \cdot k)$ is undefined, then \textbf{Game X} repeatedly chooses ${c\in_R G \setminus S^2}$ uniformly until $P^{-1}(c\cdot k^{-1})$ is undefined, i.e. the set $U=\lbrace c\in G | P^{-1}(c\cdot k^{-1})\not\in T^1\rbrace$. It follows that $y=c\cdot k^{-1}$ is uniformly distributed over $G \setminus \widehat{T}^2$.\footnote{$\widehat{T}^1$ and $\widehat{T}^2$ are the corresponding sets on the query pairs of $\widehat{P}$.} This can be seen by showing that $S^2 \cup U^\complement = \widehat{T}^2\cdot k$, where the only non-triviality in the argument follows from the definition of $\widehat{P}$. In this case, setting $E(m)=c$ also sets $\widehat{P}(m\cdot k)=y$, in contrast to the prior case as it is now defined. \end{sloppypar}
\end{itemize}
We now consider the same query on \textbf{Game X'}.
\begin{itemize}
\item[\textit{Case 1a'}:] When $\widehat{P}(m\cdot k)=P(m\cdot k)$ is defined, $c=P(m\cdot k)\cdot k$ is returned, and $\widehat{P}$ is unchanged.
\item[\textit{Case 1b'}:] When $\widehat{P}(m\cdot k)=P(m\cdot k)$ is undefined, we choose $y\in_R G\setminus T^2 = G\setminus \widehat{T}^2$, $\widehat{P}(m\cdot k)$ is set to $y$, and $c=y\cdot k$ is returned.
\end{itemize}
Thus, the behaviour of \textbf{Game X} and \textbf{Game X'} are identical on the $E$-oracle queries.

We will be briefer in our arguments for the following $3$ cases as the arguments are similar.

\textbf{Case 2: $D$-oracle query.} Here we again assume that no element of an $E/D$-oracle pair $(m,c)$, such that $E(m)=c$, has been queried before. Like in the above case, we see that, as $\widehat{P}(m \cdot k) = P(m\cdot k)$, we also have $\widehat{P}^{-1}(c \cdot k^{-1}) = P^{-1}(c\cdot k^{-1})$.
\begin{itemize}
\item[\textit{Case 2a $+$ 2a'}:] When $\widehat{P}^{-1}(c\cdot k^{-1})=P^{-1}(c\cdot k^{-1})$ is defined, then $m=P^{-1}(c\cdot k^{-1})\cdot k^{-1}$ is returned, leaving $\widehat{P}^{-1}(c \cdot k^{-1})$ unchanged in both games.
\item[\textit{Case 2b $+$ 2b'}:] If $\widehat{P}^{-1}(c\cdot k^{-1})=P^{-1}(c\cdot k^{-1})$ is undefined, then $x \in_R G \setminus \widehat{T}^1$ is chosen uniformly and $\widehat{P}^{-1}(c\cdot k^{-1}) = x$, in both cases.
\end{itemize}
Thus, the behaviour of \textbf{Game X} and \textbf{Game X'} are identical on the $D$-oracle queries.

\textbf{Case 3: $P$-oracle query.} Here we instead assume that no element of a $P/P^{-1}$-oracle pair $(x,y)$ such that $P(x)=y$, has been queried before.
\begin{itemize}
\item[\textit{Case 3a $+$ 3a'}:] Using the definition of the $E$- and $P$-oracles in \textbf{Game X} and the definition of $\widehat{P}$ we see that $P(x)$ is defined if and only if $E(x\cdot k^{-1})$ is defined, but then this also holds if and only if $\widehat{P}(x)$ is defined (by the assumption in the beginning of case $3$). Hence, if $\widehat{P}(x)$ is defined, then $y=E(x\cdot k^{-1})\cdot k^{-1}=\widehat{P}(x)$. Indeed, both games secure this value.
\item[\textit{Case 3b $+$ 3b'}:] If $\widehat{P}(x)$ is undefined, then $y \in_R G \setminus \widehat{T}^2$ is chosen uniformly and $\widehat{P}(x)$ is defined to be $y$, in both cases. 
\end{itemize}
Thus, the behaviour of \textbf{Game X} and \textbf{Game X'} are identical on the $P$-oracle queries.

\textbf{Case 4: $P^{-1}$-oracle query.} Again, we assume that no element of a $P/P^{-1}$-oracle pair $(x,y)$ such that $P(x)=y$, has been queried before.
\begin{itemize}
\item[\textit{Case 4a $+$ 4a'}:] Using the definition of \textbf{Game X} and the definition of $\widehat{P}$, as well as our case $4$ assumption, we see that $\widehat{P}^{-1}(y)$ is defined if and only if $D(y\cdot k)$ is defined. Hence, if $\widehat{P}^{-1}(y)$ is defined, then $x=D(y\cdot k) \cdot k = \widehat{P}^{-1}(y)$. Indeed, both games secure this value.
\item[\textit{Case 4b $+$ 4b'}:] If $\widehat{P}^{-1}(y)$ is undefined, then $x \in_R G \setminus \widehat{T}^1$ is chosen uniformly and $\widehat{P}^{-1}(y)$ is defined to be $x$, in both cases. 
\end{itemize}
Thus, the behaviour of \textbf{Game X} and \textbf{Game X'} are identical on the $P^{-1}$-oracle queries. Q.E.D.

\newpage
\section{Proof that the probability of Game R and Game R' match}\label{App:ReqRR}
Recall the definition of $S^1_i, S^2_i, T^1_i$ and $T^2_i$ (see p.~\pageref{GamesXandR}). We write $S_s$ and $T_t$ to denote the final transcripts. We also introduce the following definition.
\begin{defn}\label{overlapidentical}
We say that two $E/D$-pairs $( m_i,c_i )$ and $( m_j,c_j )$ \textbf{overlap} if $m_i=m_j$ or $c_i=c_j$. If $m_i=m_j$ and $c_i=c_j$, we say that the pairs are \textbf{identical}. Likewise for $P/P^{-1}$-pairs $( x_i,y_i )$ and $( x_j,y_j )$.
\end{defn}

If two pairs overlap, then by the definition of the $E/D$- and $P/P^{-1}$-oracles, they must be identical. Therefore, WLOG, we may assume that all queries to the oracles are non-overlapping. Let us now prove the lemma.

\begin{lem}
$Pr_R\left[ BAD \right] = Pr_{R'} \left[ BAD \right]$.
\end{lem}

\begin{proof}
We need to prove that \textbf{Game R} has its flag set to \textbf{bad} if and only if \textbf{Game R'} has its flag set to \textbf{bad}.
	
"$\Rightarrow$": We want to show that there exists $(m,c)\in S_s$ and $(x,y)\in T_t$ such that either $m\cdot k = x$ or $c\cdot k^{-1} = y$ (i.e. such that $k$ becomes bad). We have to consider the $8$ cases where the flag is set to bad. All of the cases use an analogous argument to the following: If $P(m \cdot k)$ is defined then $P(m\cdot k) = y = P(x)$ for some $(x,y)\in T_{t}$ such that, as overlapping pairs are identical, $m\cdot k = x$.

"$\Leftarrow$": We assume that there exists $(m,c)\in S_s$ and $(x,y)\in T_t$  such that $k$ becomes bad. i.e. such that either $m\cdot k = x$ or $c\cdot k^{-1} = y$. We need to check that in all four oracle queries, the flag in \textbf{Game R} is set to bad, which needs a consideration of 8 cases.

Assume that $m\cdot k = x$, then
\begin{align*}
E\text{-oracle on } m &: P(m\cdot k)=P(x) = y \in T_t^2, \\
D\text{-oracle on } c &: P(m\cdot k)=P(x) = y \in T_t^2, \\
P\text{-oracle on } x &: E(x\cdot k^{-1}) = E(m) = c \in S_s^2, \\
P^{-1}\text{-oracle on } y &: E(x\cdot k^{-1}) = E(m) = c \in S_s^2.
\end{align*}

Assume now that $c\cdot k^{-1} = y$, then
\begin{align*}
E\text{-oracle on } m &: P^{-1}(c\cdot k^{-1})=P^{-1}(y) = x \in T_t^1, \\
D\text{-oracle on } c &: P^{-1}(c\cdot k^{-1})=P^{-1}(y) = x \in T_t^1, \\
P\text{-oracle on } x &: D(y\cdot k) = D(c) = m \in S_s^1, \\
P^{-1}\text{-oracle on } y &: D(y\cdot k) = D(c) = m \in S_s^1.
\end{align*}
\end{proof}

\end{document}